\documentclass[12pt,draftclsnofoot,onecolumn]{IEEEtran}
\usepackage{bm}
\usepackage{amsmath}
\usepackage{booktabs}
\usepackage{amssymb}
\usepackage{amsthm}
\usepackage{graphicx}
\usepackage{array}
\usepackage{nameref}
\usepackage{lipsum}

\usepackage{multirow}
\newtheorem{theorem}{Theorem}

\newtheorem{lemma}[theorem]{Lemma}
\newtheorem{auxiliary code}{Auxiliary Code}

\usepackage{subcaption}

\usepackage{authblk}
\usepackage[T1]{fontenc}
\usepackage{stfloats}
\usepackage{url}

\usepackage{color}
\usepackage{nccmath}
\usepackage[ruled]{algorithm2e}
\begin{document}
\title{Channel Models and Coding Solutions for 1S1R Crossbar Resistive Memory with High Line Resistance}

\author{Zehui~Chen,~\IEEEmembership{Student Member,~IEEE,}
	and~Lara~Dolecek,~\IEEEmembership{Senior Member,~IEEE}
	\thanks{Z. Chen, and L. Dolecek are with the Electrical and Computer Engineering Department, University of California, Los Angeles, Los Angeles, CA 90095, USA. (email: chen1046@ucla.edu; dolecek@ee.ucla.edu). 
	This research is supported in part by a grant from UC MEXUS and an NSF-BSF grant no.1718389. Partial results in this paper were presented at the IEEE Global Communication Conference, Taipei, China, Dec. 2020 (reference \cite{chen2020}).}}
\maketitle

\begin{abstract}
	Crossbar resistive memory with the 1 Selector 1 Resistor (1S1R) structure is attractive for nonvolatile, high-density, and low-latency storage-class memory applications. As technology scales down to the single-nm regime, the increasing resistivity of wordline/bitline becomes a limiting factor to device reliability. This paper presents write/read communication channels while considering the line resistance and device variabilities by statistically relating the degraded write/read margins and the channel parameters. Binary asymmetric channel (BAC) models are proposed for the write/read operations. Simulations based on these models suggest that the bit-error rate of devices are highly non-uniform across the memory array. These models provide quantitative tools for evaluating the trade-offs between memory reliability and design parameters, such as array size, technology nodes, and aspect ratio, and also for designing coding-theoretic solutions that would be most effective for crossbar memory. Method for optimizing the read threshold is proposed to reduce the raw bit-error rate (RBER). We propose two schemes for efficient channel coding based on Bose-Chaudhuri-Hocquenghem (BCH) codes. An interleaved coding scheme is proposed to mitigate the non-uniformity of reliability and a location dependent coding framework is proposed to leverage this non-uniformity.  Both of our proposed coding schemes effectively reduce the undetected bit-error rate (UBER).  
\end{abstract}

\section{Introduction}
The crossbar resistive memory, whereby bistable memristors are placed at the crosspoint of wordlines and bitlines, is one promising candidate for the next generation nonvolatile memory due to its inherent $4F^2$ device density and its simple crossbar structure \cite{ielmini2015resistive}. One of its most appealing applications is storage-class memory (SCM), a term that refers to memory technology that fills the latency and density gaps between DRAM and NAND flash memory \cite{burr2008overview}. Meanwhile, as semiconductor technology scales down to single-digit-nm, simultaneously scaled wordline/bitline resistances increasingly become a limiting factor to device reliability and hence memory scalability \cite{liang2013effect}. 

Previous literature has extensively shown that even moderate line resistance significantly degrades the reliability of the write and read operations. The degradation of the write/read margins due to high line resistance for the wort-case memory cell, i.e., the cell that is furthest from the source and ground, are studied in \cite{liang2013effect,chen2016design,kim2015numerical}. The adverse effect of the line resistance on the write/read margins for cells across the memory array are studied in \cite{chen2013comprehensive,shin2010data} by solving a system of Kirchhoff's current law (KCL) equations. While these studies focused on the degradation of the write/read margins, it remains unclear how the degraded write/read margins affect the system level reliability metric, e.g., the bit-error rate (BER). In other words, channel models are not yet well-established for this problem.

It is demonstrated in \cite{chen2013comprehensive} that, when considering the line resistance in resistive memory, the write margins are non-uniform across the array, which leads to non-uniform reliability levels in the memory array. Designing error correction codes (ECCs) for the worst-case often leads to overly conservative code design and is therefore not rate efficient. For example, in \cite{zorgui2019polar}, the authors designed a non-stationary polar code targeting channels with different reliability levels, which are characterized empirically by simulations. Moreover, \cite{zorgui2019polar} also showed that using more precise channel modeling, i.e., using the binary asymmetric channel (BAC) instead of the binary symmetric channel (BSC), provides an order of magnitude improvement in BER, which proves the necessity of precise channel models. 

In this work, we propose BAC models for writing to and reading from memory devices in crossbar memory, parameterized by device parameters, array size, wordline/bitline resistances, and device location by statistically relating the degraded write/read margins of cells at different locations to the channel parameters. Our analytical channel models, which take into account the device location, provide quantitative tools for analyzing the aforementioned non-uniformity and the trade-off between device parameters and memory reliability. These models are therefore beneficial for system engineers when designing the next generation storage systems. Previous studies on the write/read margins assume deterministic High Resistance State (HRS) and Low Resistance State (LRS) for the memory device whereas the HRS and LRS are nondeterministic in nature \cite{ji2015line,chen2011variability}. Our write/read channel models, which are derived probabilistically, allow us to take the resistance variability of LRS and HRS into consideration for more precise modeling. 

Building upon our proposed channel models and the observation of non-uniform reliability, we propose methods to reduce the raw bit-error rate (RBER) and undetected bit-error rate (UBER) using techniques from estimation theory and channel coding theory. For the read channel, we propose an efficient procedure to compute an optimal read threshold. We show that the RBER of the read channel is reduced by a large factor using the optimal read threshold. Efficient ECC solution for a crossbar resistive memory targeting the SCM application must mitigate and/or leverage the non-uniformity of reliability while being compatible with the low-latency requirement of SCM. Based on BCH codes, we propose a scheme based on interleaving and a scheme utilizing multiple BCH codes with different error correction capabilities for improved UBER performance. For the latter scheme, we propose a systematic framework for allocating different codes based on the location dependent channel parameters. Both of our proposed coding schemes effectively reduce the UBER.

The content of this paper is organized as follows. Section II provides background on crossbar resistive memory and the write/read operation. The circuit models and the variabilities are also discussed in Section II. Section III presents the channel characterization for the write and read operations. Simulation results on the proposed channel models with various parameters are also presented in Section III. Section IV presents our study on the optimal read threshold and its simulation results. Section V discusses two schemes for efficient channel coding targeting SCM applications. We conclude and discuss future research in Section VI.

\section{Preliminaries}
\subsection{1S1R Crossbar resistive memory Background and Model}
In crossbar resistive memory array, the logical state $0$ or $1$ is represented by the HRS or LRS of a memory cell, respectively. For a bipolar memristor, the state of a cell is switched from LRS to HRS (Reset Operation) or from HRS to LRS (Set Operation) by applying a positive or negative voltage across the memory cell, respectively. For the write operation, we consider the so called ``V/2'' write scheme (cf. \cite{chen2016design}) as it is usually more energy-efficient than the so called ``V/3'' write scheme. In particular, when writing to a selected cell, the wordline and bitline of the selected cell are biased at the write voltage ($V_{w\_set}$ or $V_{w\_reset}$) and 0, respectively, while other wordlines and bitlines are biased at half of the write voltage to prevent unintentional write, as shown in Fig. \ref{fig:write model}. 

\begin{figure}[h]
	\centering
	\begin{subfigure}{0.49\textwidth}
	\centering
	\includegraphics[scale=0.43]{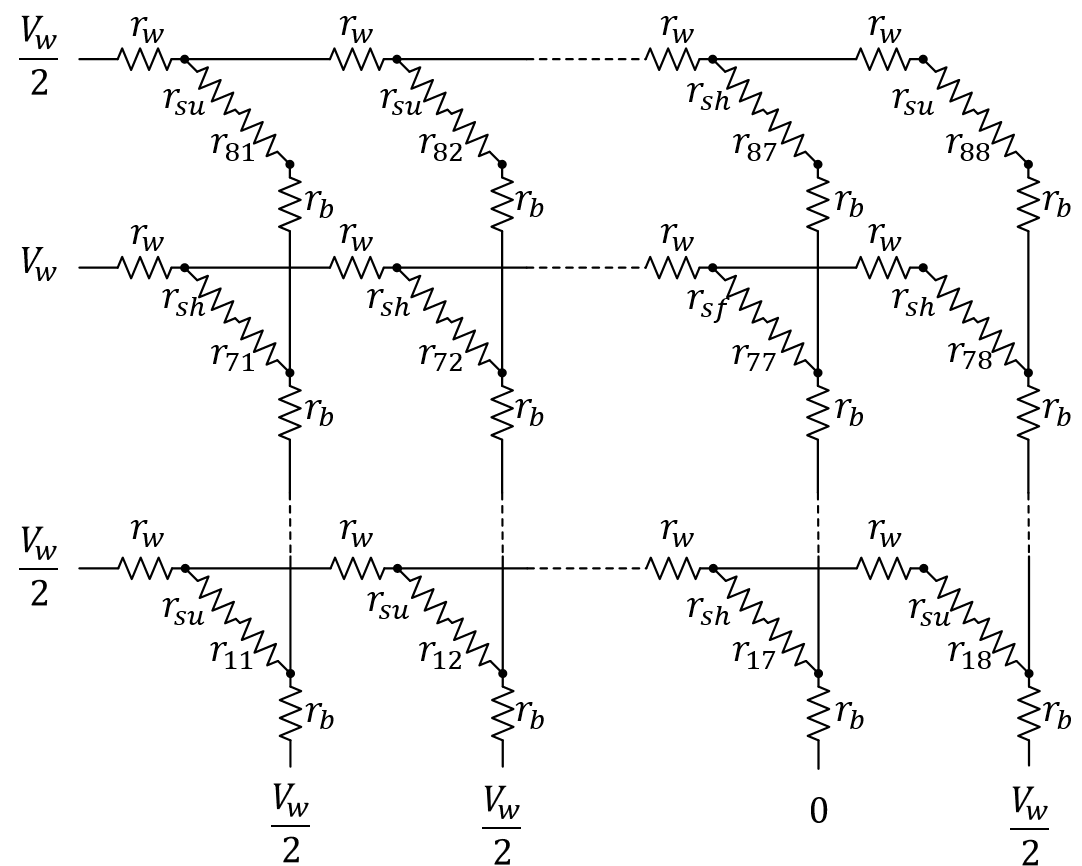}	
	\caption{Circuit model for writing to a $8\times8$ array.}
	\label{fig:write model}
	\end{subfigure}
	\begin{subfigure}{0.49\textwidth}
	\centering
	\includegraphics[scale=0.445]{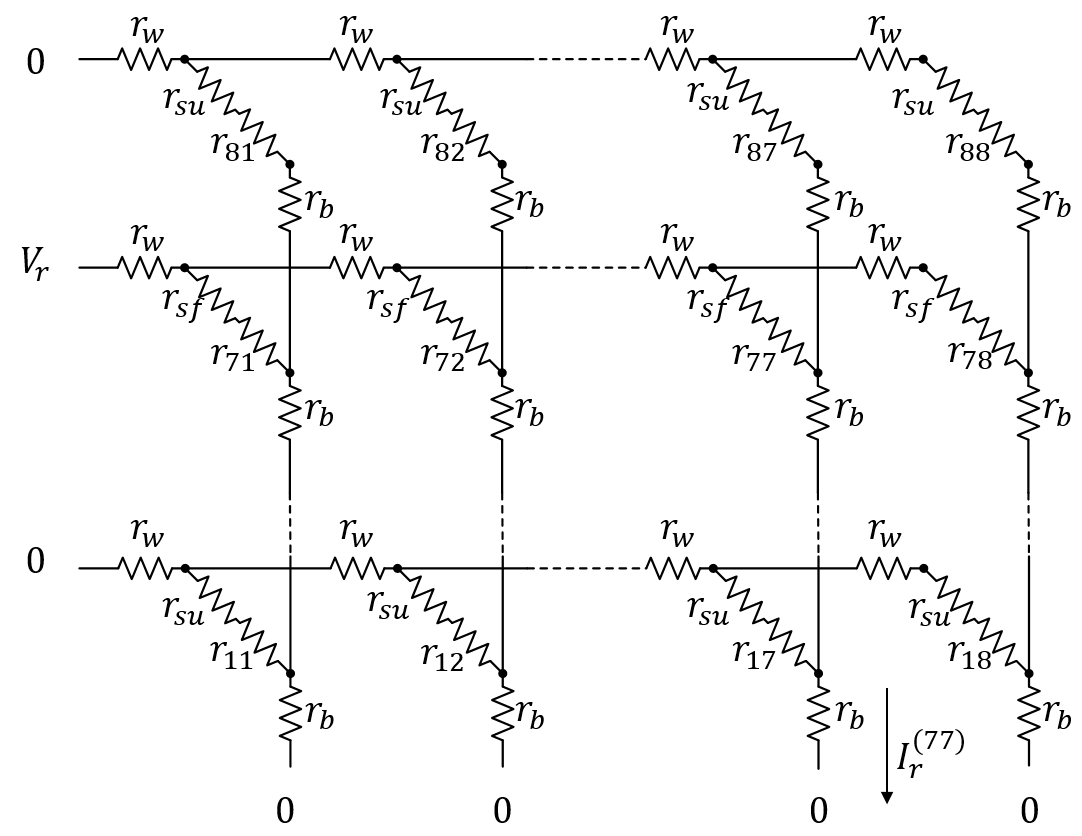}	
	\caption{Circuit model for reading to a $8\times8$ array.}
	\label{fig:read model}
	\end{subfigure}
	\caption{Examples of Circuit Models ($V_w$ denotes $V_{w\_set}$ or $V_{w\_reset}$).}
	\vspace{-1em}
\end{figure}

For the read operation, we consider the current-mode sensing scheme as it has a smaller latency compared with the voltage-mode sensing scheme \cite{chen2016design}. When reading a selected cell, a read voltage ($V_r$) is applied on its wordline and all other wordlines and bitlines are grounded. A current is sensed by the sensing amplifier located at the end of its bitline, and is used to determine the state of the selected cell, as shown in Fig. \ref{fig:read model}.

In this paper, we focus on crossbar resistive memory with the widely used 1 selector 1 resistor (1S1R) structure, where highly nonlinear selectors are connected in series with the memristors to prevent write and read disturbs. For both write and read operations, when the voltage across a selector is close to the applied voltage, we say that this selector is fully selected and we assume it has resistance $r_{sf}$; when the voltage across a selector is close to 0, we say that this selector is un-selected and we assume it has resistance $r_{su}$. For the write operation, since other cells on the wordline and bitline of the selected cell have voltage close to half of the write voltage across them, we say that the selectors for those cells are half-selected and we assume they have resistance $r_{sh}$. In general, $r_{sf}<<r_{sh}<r_{rs}$. An ideal selector has  parameters $r_{sf} = 0$ and $r_{sh}=r_{su}=\infty$. Our proposed model is a general one that does not have the ideal selector assumption. Meanwhile, since the main focus of this work is the adverse effect of line resistance, we use the ideal selector assumption to provide mathematical insights in III.B and to simplify our simulations in III.D and IV.C. Throughout this work, we assume that the interconnect resistances of wordlines and bitlines are constant across the array, and they are denoted by $r_w$ and $r_b$ respectively.
\subsection{Memristor Variabilities and Models}
In this paper, we consider two variabilities of memristor, the non-deterministic write operation and the non-deterministic resistance value for each resistance state.  It is widely observed that the switching operations of memristor are stochastic and follow log-normal switching time distributions, with distribution parameters depend on the applied voltage \cite{medeiros2011lognormal,niu2012low}. Our models for the switching time distributions are adopted from \cite{medeiros2011lognormal} and more details are provided in Section III. 

Previous works (cf. \cite{liang2013effect} - \cite{shin2010data}) on the degradation of write and read margins due to high line resistance assume deterministic resistance states, e.g., HRS resistance is $10000\Omega$ and LRS resistance is $100\Omega$. Meanwhile, due to both device-to-device variation and cycle-to-cycle variation, the resistance of each state is highly non-deterministic \cite{ji2015line,chen2011variability}. To incorporate this variability into our reliability analysis, we use random variables to represent the resistance of the memory cells. Based on observations in \cite{ji2015line,chen2011variability}, we assume they are i.i.d. and their conditional distributions, conditioned on their states, follow log-normal distributions.  For example, let i.i.d. Bernoulli($q$) random variable $S_{ij}$ denote the state of cell $(i,j)$, with $S_{ij} = 1$ for LRS and $S_{ij} = 0$ for HRS. Let $R_{ij}$ be the associated random variable denoting the resistance of cell $(i,j)$. Then our model assumes:
\[\ln(R_{ij}|S_{ij}=1)\sim \mathcal{N}(\mu_{L},\sigma_{L}^2),\]
and
\[\ln(R_{ij}|S_{ij}=0)\sim \mathcal{N}(\mu_{H},\sigma_{H}^2).\]

\section{Channel Models}
Our proposed channel models are depicted in Fig. \ref{fig:cascaded_channel}. We model the write, read and cascaded channels as binary asymmetric channels (BACs). We discuss in the following subsections how the channel parameters are related to device parameters, line resistance, and device location. 
\begin{figure}[h]
	\centering
	\includegraphics[scale=0.37]{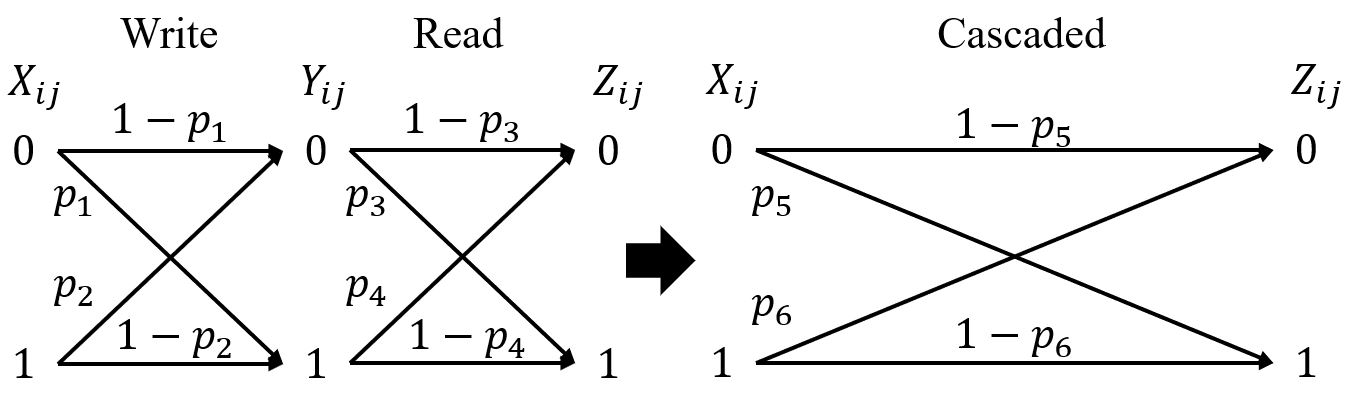}	
	\caption{Proposed Channel Models.}
	\label{fig:cascaded_channel}
	\vspace{-1em}
\end{figure}

We denote the state we want to write to cell $(i,j)$ by $X_{ij}\in\{0,1\}$, the state actually written by $Y_{ij}\in\{0,1\}$ and the detected (for read operation) state by $Z_{ij}\in\{0,1\}$. Note that even though the information is stored as the resistance of a cell, we choose to use binarized state variable $Y_{ij}$ because firstly it {enables} us to utilize a well-known result in the literature \cite{medeiros2011lognormal} that characterizes the switching of a device, and secondly it allows mathematical tractability and the separation of the write/read channels. Information about the resistance of a cell is instead embedded in the resistance distribution. Also note that with a binarized state for a cell, multiple write/read operations are also independent. Therefore, the cascaded channel model is still valid if one writes to and reads from a cell multiple times.
\subsection{Write Channel}
In this section, we derive the write channel. We note that the write operation is affected by the previous state of cell $(i,j)$. We let this state be denoted by $S^*_{ij}$ and the associated resistance value be $R^*_{ij}$. We assume that when the previous state is the same as the state we want to write, the write operation is always successful, i.e., 
$P(Y_{ij}=1|X_{ij}=1,S^*_{ij}=1)=1,$
and $P(Y_{ij}=0|X_{ij}=0,S^*_{ij}=0)=1.$

When the previous state is not the same as the state we want to write, a sufficient write voltage and a sufficient write time {are} required to change the state of the cell. Due to high line resistances, the effective write voltage on a cell could be much smaller than the desired write voltage, i.e., the write margin is decreased. We denote the effective write voltage on a cell $(i,j)$ as $\tilde{V}_w(r^*_{ij},i,j)$ where $r^*_{ij}$ is a realization of $R^*_{ij}$. With a method similar to the one described in \cite{chen2013comprehensive}, $\tilde{V}_w(r^*_{ij},i,j)$ can be obtained by solving a system of KCL equations using the circuit model described in Section II.A. We map the degraded write margin to the decreased write reliability by considering the log-normal switching time distribution, adopted from \cite{medeiros2011lognormal}. With fixed switching times $t_{set}$ and $t_{reset}$, the log-normal switching time distributions lead to the following:
\begin{equation}
\label{equ:set_success}
P(Y_{ij}=1|X_{ij}=1,S^*_{ij}=0,R^*_{ij}=r^*_{ij})=1-Q\left(\frac{\ln t_{set} -\ln(\tau^{(ij)}_{set})}{\sigma_{set}}\right),
\end{equation} 
and
\begin{equation}
\label{equ:reset_success}
P(Y_{ij}=0|X_{ij}=0,S^*_{ij}=1,R^*_{ij}=r^*_{ij})
=1-Q\left(\frac{\ln t_{reset} -\ln(\tau^{(ij)}_{reset})}{\sigma_{reset}}\right),
\end{equation} 
where $Q(\cdot)$ is the $Q$-function, i.e., $Q(x) = \frac{1}{\sqrt{2\pi}}\int_{x}^{\infty}\exp(-\frac{u^2}{2})du$. Parameters $\sigma^2_{set}$ and $\sigma^2_{reset}$ are the variance of the normal distributions associated with the set and reset switching time distribution, respectively, which are independent of $\tilde{V}_w(r^*_{ij},i,j)$, according to \cite{medeiros2011lognormal}. Parameters $\tau^{(ij)}_{set}$ and $\tau^{(ij)}_{reset}$ are the median of the set and reset switching time, respectively. Note that in the above equations, to be consistent with the existing literature \cite{medeiros2011lognormal,niu2012low}, we use the median parameterization of the log-normal distribution. According to the literature, the medians of the switching times ($\tau^{(ij)}_{set}$ and $\tau^{(ij)}_{reset}$ in $\mu s$) are exponentially dependent on the effective write voltage. We therefore parameterize the medians as following:
\[\ln\left(\tau^{(ij)}_{set}\right) = \alpha_{set}\tilde{V}_w(r^*_{ij},i,j)+\beta_{set},\]
and
\[\ln\left(\tau^{(ij)}_{reset}\right) = \alpha_{reset}\tilde{V}_w(r^*_{ij},i,j)+\beta_{reset}.\]

Using (\ref{equ:set_success}), (\ref{equ:reset_success}) and marginalizing over the conditionally log-normally distributed random variable $R^*_{ij}$, we get:
\small
\begin{equation}
\label{equ:set_fail}
P(Y_{ij}=0|X_{ij}=1,S^*_{ij}=0)=\int_{-\infty}^{\infty}\frac{1}{\sqrt{2\pi}r^*_{ij}\sigma_{H}} \exp\left[-\frac{\left(\ln r^*_{ij}-\mu_H\right)^2}{2\sigma_{H}^2}\right]Q\left(\frac{\ln t_{set} -\ln(\tau^{(ij)}_{set})}{\sigma_{set}}\right)dr^*_{ij},
\end{equation} 
and
\begin{equation}
\label{equ:reset_fail}
P(Y_{ij}=1|X_{ij}=0,S^*_{ij}=1)=\int_{-\infty}^{\infty}\frac{1}{\sqrt{2\pi}r^*_{ij}\sigma_{L}}\exp\left[-\frac{\left(\ln r^*_{ij}-\mu_L\right)^2}{2\sigma_{L}^2}\right]Q\left(\frac{\ln t_{reset} -\ln(\tau^{(ij)}_{reset})}{\sigma_{reset}}\right)dr^*_{ij}.
\end{equation} 
\normalsize
Putting (\ref{equ:set_fail}) and (\ref{equ:reset_fail}) together with the prior symbol probability $q=P(S^*_{ij}=0)$, we arrive at the write binary asymmetric channel, depicted in Fig. \ref{fig:cascaded_channel}, for the write operation with the following channel parameters:
\begin{equation}
\label{equ:p1}
p^{(ij)}_1 = (1-q)P(Y_{ij}=1|X_{ij}=0,S^*_{ij}=1),
\end{equation}
and
\begin{equation}
\label{equ:p2}
p^{(ij)}_2 = qP(Y_{ij}=0|X_{ij}=1,S^*_{ij}=0).
\end{equation}
Here and elsewhere, we use superscript $(ij)$ to highlight that the channel parameters are dependent on the cell location $(i,j)$. 

Through equations (\ref{equ:set_success}) - (\ref{equ:p2}), we are able to relate the write margin $\tilde{V}_w(r^*_{ij},i,j)$ to the BER of the write channel. For example, comparing the best-case cell to the worst-case cell in the example in Section III.D Fig. \ref{fig:heat_map}, we observe that the write margin for Reset is dropped from $4.9 V$ to $1.64 V$ while the write BER is increased from $3.35\times10^{-4}$ to $1.75\times10^{-2}$, thus providing further evidence that location dependent BER analysis matters.

\subsection{Read Channel}
When reading from the cell $(i.j)$, we consider the current-mode sensing scheme and a fixed threshold detector. Let $I^{(ij)}_r$ be the current sensed by the sensing amplifier, which can be also calculated by solving a system of KCL equations. $I^{(ij)}_r$ is hence dependent on the cell location, the resistance of the selected cell, and the resistances of unselected cells. Let $I_{th}$ be the threshold current. The threshold detector is as follows:
\begin{equation}
Z_{ij}=\begin{cases}
0, I^{(ij)}_r\leq I_{th},\\
1, I^{(ij)}_r>I_{th}.
\end{cases}
\end{equation}
With the threshold detector above, the decision error probabilities are:
\begin{equation}
\label{equ:0_read_error}
P(Z_{ij}=1|Y_{ij}=0) = P(I^{(ij)}_r> I_{th}|Y_{ij}=0);
\end{equation}

\begin{equation}
\label{equ:1_read_error}
P(Z_{ij}=0|Y_{ij}=1) = P(I^{(ij)}_r\leq I_{th}|Y_{ij}=1).
\end{equation}
This set-up leads to the read binary asymmetric channel, depicted in Fig. \ref{fig:cascaded_channel}, for the read operation with $p^{(ij)}_3 = P(Z_{ij}=1|Y_{ij}=0)$ and $p^{(ij)}_4 = P(Z_{ij}=0|Y_{ij}=1)$.

\subsubsection{Closed form Expression with Ideal Selectors}
Since we need to solve a system of equations to get $I^{(ij)}_r$, equations (\ref{equ:0_read_error}) and (\ref{equ:1_read_error}) are not sufficient as they do not give closed-form expressions for the channel parameters. However, if we consider ideal selectors, closed-form expressions can be derived. { Note that for the analysis with a non-ideal selector, one can still use the general characterizations in equations (\ref{equ:0_read_error}) and (\ref{equ:1_read_error}) and find $I^{(ij)}_r$ by solving a system of KCL equations. Moreover, it is reasonable to assume ideal characteristics of the selector for the analysis of the line resistance in the 1S1R structure as near ideal selector properties are demonstrated by industry in  \cite{jo20143d}.}

With ideal selectors, the part of the circuit connected to the un-selected cells can be neglected, resulting in a simplified circuit with just the selected cell and its wordline/bitline. With this simplified circuit, $I^{(ij)}_r$ is a function of the random variable $R_{ij}$, which represents the resistance of the selected cell. We therefore have:
\begin{equation}
\label{equ:I_r}
I^{(ij)}_r = \frac{V_r}{ir_b+jr_w+R_{ij}}.
\end{equation}
Plugging (\ref{equ:I_r}) into (\ref{equ:0_read_error}) and (\ref{equ:1_read_error}), and using the assumption that $R_{ij}$ is conditionally (on $Y_{ij}$) log-normally distributed, we obtain the following closed form expression for $p_3$ and $p_4$:
\begin{equation}
\label{equ:p3}
\begin{split}
p^{(ij)}_3 &= P\left(\frac{V_r}{ir_b+jr_w+R_{ij}}>I_{th}|Y_{ij}=0\right)=P\left(R_{ij}<\frac{V_r}{I_{th}}-ir_b-jr_w|Y_{ij}=0\right)\\
&=Q\left(\frac{\mu_{H}-\ln\left(\frac{V_r}{I_{th}}-ir_b-jr_w\right)}{\sigma_{H}}\right),
\end{split}
\end{equation}
and similarly
\begin{equation}
\label{equ:p4}
\begin{split}
p^{(ij)}_4 =Q\left(\frac{\ln\left(\frac{V_r}{I_{th}}-ir_b-jr_w\right)-\mu_{L}}{\sigma_{L}}\right).
\end{split}
\end{equation}

Define $R_{th} = \frac{V_r}{I_{th}}$. From equations (\ref{equ:p3}) and (\ref{equ:p4}), we observe that $R_{th}$ is the effective decision threshold between the HRS and LRS distribution in the resistance domain, when there {is} no line resistance, i.e., $r_w=r_b=0$. We can therefore interpret the adverse effect of line resistances during the read operation as follows: the effective read threshold in resistance domain is shifted to the left by the total accumulated line resistance. This shift results in a higher bit-error rate if $R_{th}$ is set to be the optimal decision threshold without considering the line resistance. 

The read margin is defined by the difference between the sensed current of a HRS cell and the sensed current of a LRS cell. Using equations (\ref{equ:I_r}) - (\ref{equ:p4}), we can now relate the read margin to the read BER. For example, comparing the best-case cell to the worst-case cell in the example in Section III.D Fig. \ref{fig:heat_map}, we observe that the read margin is dropped from $296 \mu A$ to $95 \mu A$ while the write BER is increased from $4.29\times10^{-4}$ to $7.33\times10^{-2}$, again demonstrating the need of a location dependent model.

\subsection{Cascaded Channel and Channel Capacity}
Combining the results of the previous two subsections, we get a cascaded channel for a single memory cell. The cascaded channel is a binary asymmetric channel and it is depicted in Fig. \ref{fig:cascaded_channel}, with $p^{(ij)}_5 = p^{(ij)}_1(1-p^{(ij)}_4)+(1-p^{(ij)}_1)p^{(ij)}_3$ and $p^{(ij)}_6 = p^{(ij)}_2(1-p^{(ij)}_3)+(1-p^{(ij)}_2)p^{(ij)}_4$.

The capacity of this cascaded channel for cell $(i,j)$ is as follows:
\begin{equation}
\label{equ:cell_capacity}
\begin{split}
C_{ij} &= \max_{q}\,I(X_{ij};Z_{ij})\\
&=\max_{q}\,\Bigg[h\left(q\left(1-p^{(ij)}_5\right)+(1-q)p^{(ij)}_6\right)-qh\left(p^{(ij)}_5\right)-(1-q)h\left(p^{(ij)}_6\right)\Bigg],
\end{split}
\end{equation}
where $h(\cdot)$ is the binary entropy function. Because $p^{(ij)}_1$ and $p^{(ij)}_2$ are dependent on $q$, the closed form capacity result for a standard BAC does not hold. The channel capacity therefore need to be evaluated with a numerical method, e.g., the Blahut-Arimoto algorithm, as further presented in Subsection D. 

\subsection{Simulations Results}
Based on our models presented in the previous subsections, we simulate multiple arrays to explore how memory parameters affect the memory reliability metrics, such as the bit-error rate (BER) and the averaged capacity. We calculate the averaged capacity by averaging the capacities of cells given by equation (\ref{equ:cell_capacity}); this result serves as an indicator of what fraction of the input data can be reliably stored in memory. Since this work is mainly focused on the adverse effect of the line resistance, we only vary the array size, aspect ratio, and line resistance in our simulations. Other memory parameters are kept the same and are summarized in Table \ref{table1}. As an illustrative example, the parameters are chosen to represent a moderate reliability level, with a BER on the order of $10^{-3}$ in the best case scenario. { The considered line resistances range from $10\Omega$ to $100\Omega$, in accordance with the interconnect resistance values of interest for moderate technology nodes \cite{liang2013effect}.  The chosen standard deviation (0.3) of LRS and HRS distribution is experimentally observed in \cite{ji2015line}. The chosen switching parameters are based on \cite{medeiros2011lognormal}; we use the same parameters for reset and set operations for simplified analysis.
}
\begin{table}[h!]
	\centering
	\renewcommand{\arraystretch}{1}
	\begin{tabular}{c|c|c}
		\hline
		\hline
		Symbol 		& Parameters 						& Values \\ 
		$m,n$			& Array Size ($m\times n$)          &   varies               \\ 
		$V_{w\_set}$			& Set voltage       &    -5V              \\ 
		$V_{w\_reset}$			& Reset voltage       &    5V              \\ 
		$V_{r}$			& Read voltage       &    3V              \\ 
		$q$		& Prior symbol probability of 0          & $0.5$                 \\ 
		$r_w$		&Wordline interconnect resistance       &   $10\Omega-100\Omega$               \\ 
		$r_b$		&Bitline interconnect resistance       &    $10\Omega-100\Omega$              \\ 
		$r_{sf}$		&Fully selected selector resistance    &         $0$       \\ 
		$r_{sh}$		&Half selected selector resistance    &         $\infty$       \\ 
		$r_{su}$		&Unselected selector resistance   &         $\infty$       \\ 
		$\mu_L$ 	& Associated mean of LRS distribution     	&    $4\ln(10)$              \\
		$\mu_H$ 	& Associated mean of HRS distribution      	&    $6\ln(10)$              \\
		$\sigma_L$ 	& Associated std of LRS distribution       	&      $0.3\ln(10)$            \\
		$\sigma_H$ 	& Associated std of HRS distribution       	&      $0.3\ln(10)$           \\
		$\alpha_{set}$ 	& Parameter for the median set time    	&      $0.25$            \\
		$\beta_{set}$ 	& Parameter for the median set time    	&       $4.25$           \\
		$\alpha_{reset}$ 	& Parameter for the median reset time    	&    $-0.25$              \\
		$\beta_{reset}$ 	& Parameter for the median reset time    	&     $4.25$             \\
		$\sigma_{set}$ 	& Associated std of set time distribution       	&        $0.5$          \\
		$\sigma_{reset}$ 	& Associated std of reset time distribution       	&        $0.5$          \\
		$t_{set}$ 	&  Switching time for set operation      	&    $100\mu s$              \\
		$t_{reset}$ 	&  Switching time for reset operation      	&      $100\mu s$            \\
		$I_{th}$ 	& Read decision threshold       	&      $30\mu A$            \\
		\hline
		\hline
	\end{tabular}
	\caption{Summary of Parameters.}
	\label{table1}
	\vspace{-1em}
\end{table}

\begin{figure}[h]
	\centering
	\begin{subfigure}{0.7\textwidth}
		\centering
		\includegraphics[scale=0.3]{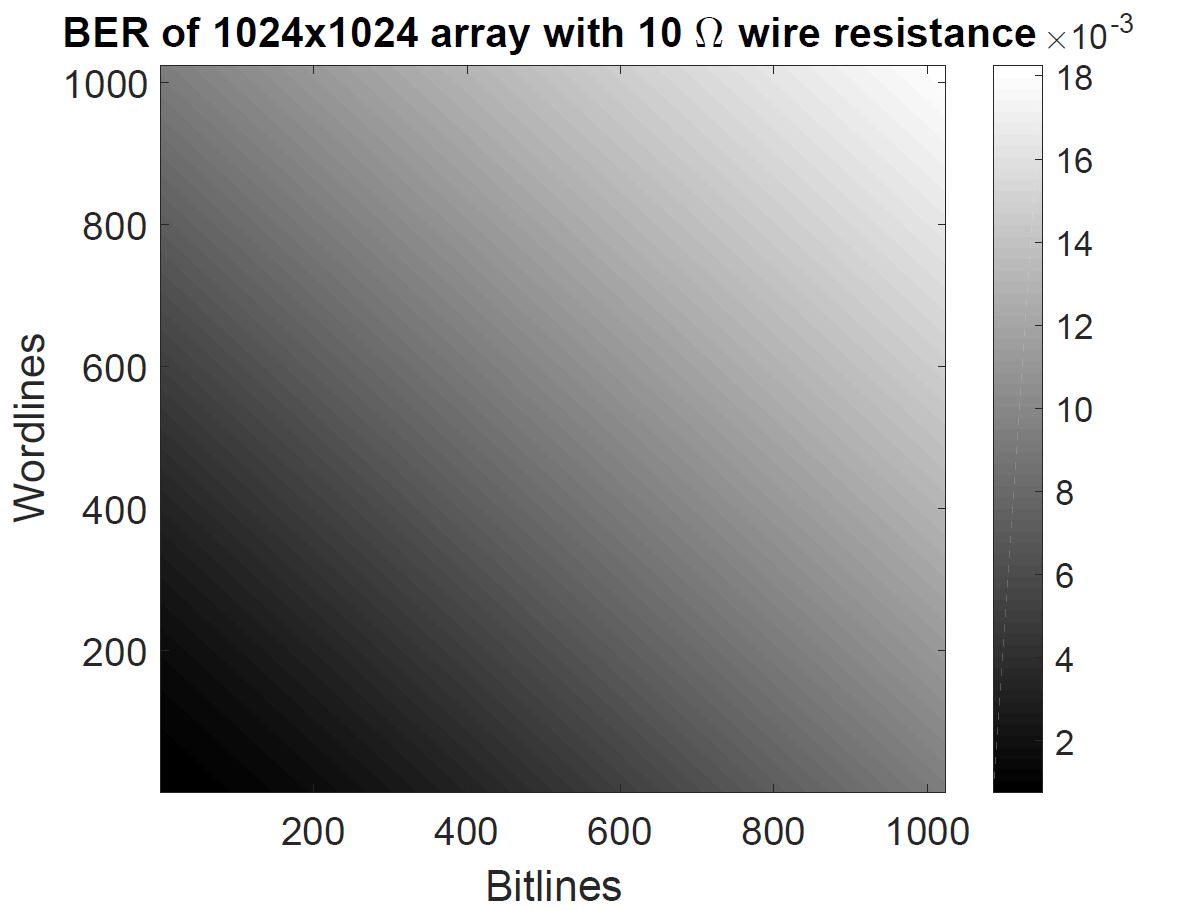}	
		\caption{Heatmap of BERs for a 1024x1024 array.}
		\label{fig:heat_map}
	\end{subfigure}
	\begin{subfigure}{0.49\textwidth}
		\centering
		\includegraphics[scale=0.35]{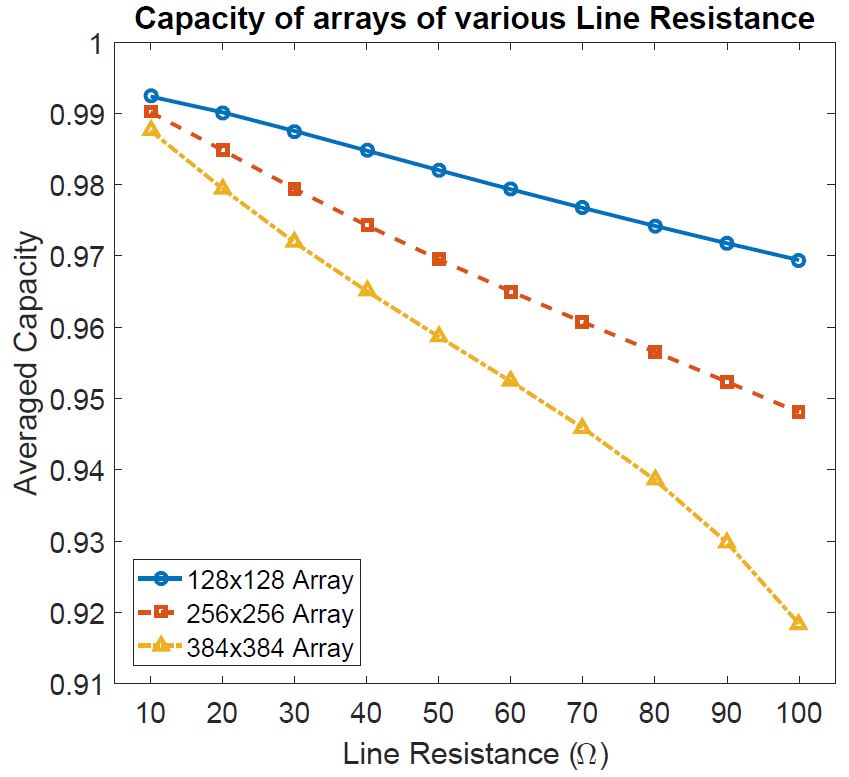}	
		\caption{Capacity results of various line resistances.}
		\label{fig:Capacity_vs_rl}
	\end{subfigure}
	\begin{subfigure}{0.49\textwidth}
	\centering
	\includegraphics[scale=0.35]{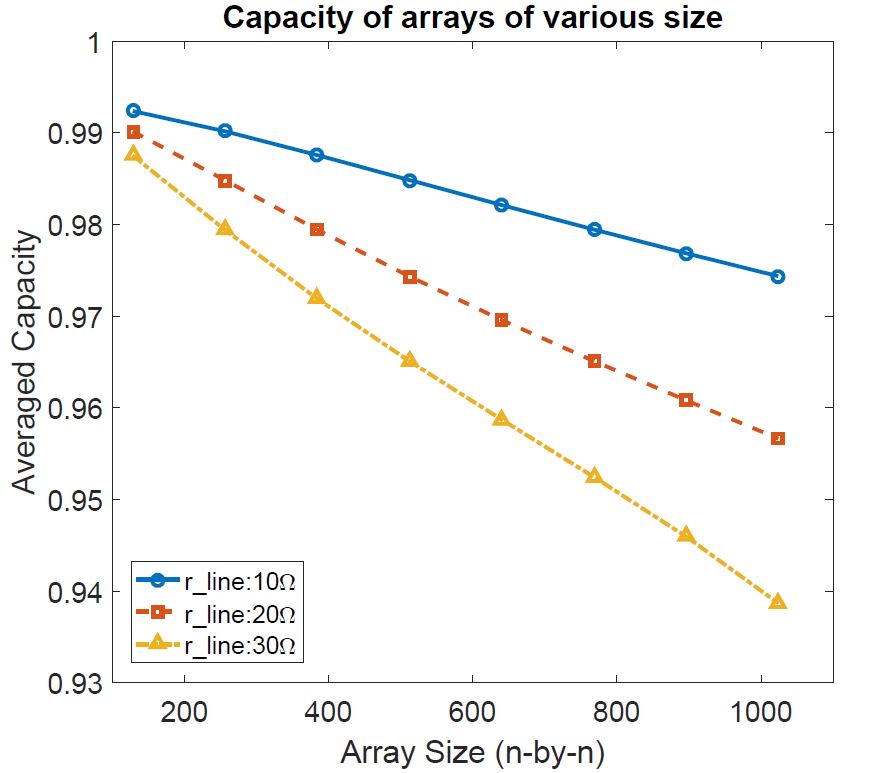}	
	\caption{Capacity results of various array sizes.}
	\label{fig:Capacity_vs_N}
\end{subfigure}
	\caption{Simulation Results for the Proposed Channel Models.}
\end{figure}
In Fig. \ref{fig:heat_map}, we first present the BER of each cell in a $1024\times1024$ array to illustrate the spatial variation of reliability due to the line resistance. According to \cite{liang2013effect}, the chosen $10\Omega$ line resistance corresponds to the resistance per junction of Cu wire with 20nm technology nodes. With this moderate line resistance, we observe an order of magnitude BER difference between the best-case cell, located closest to the voltage source, and the worst-case cell, located furthest from the voltage source. Due to line resistance, the cell which is further from the source and sensing amplifier, suffers from a lower voltage delivery during the write operation and a higher resistance interference during the read operation, thus has a larger BER. 

Next, in Fig. \ref{fig:Capacity_vs_rl} and Fig. \ref{fig:Capacity_vs_N}, we present the averaged capacity per cell for arrays with various size and line resistances, with aspect ratio fixed to be $1$. We observe that a larger line resistance, which corresponds to a smaller technology node, deteriorates the averaged capacity almost linearly. This trade-off thus must be taken into consideration when scaling the memory, as it is shown in \cite{liang2013effect} that the line resistance scales exponentially with respect to the technology node. Also note that when the accumulated line resistance of the worst-case cell gets close to the effective resistance threshold, i.e., when $nr_w+mr_b$ is close to $R_{th}$, the averaged capacity deteriorates faster, as from (\ref{equ:p4}), when $nr_w+mr_b>R_{th}$, reading from a LRS cell correctly is impossible. This explains the rapid dropping at the end of the curve in Fig. \ref{fig:Capacity_vs_rl} for the $384\times384$ array. From  Fig. \ref{fig:Capacity_vs_N}, we notice that the averaged capacity also deteriorates almost linearly with respect to the array size. This effect is thus a limiting factor for the realization of a large memory array.

\begin{table}[h!]
	\renewcommand{\arraystretch}{1.1}
	\centering
	\begin{tabular}{|c|c|c|c|c|c|c|}
		\hline
		Array Size & $128\times128$ &$64\times512$  &$32\times512$  & $16\times1024$ &$8\times2048$  &$4\times4096$ \\ \hline
		Averaged Capacity   & $0.9924$ & $0.9918$ & $0.9897$  & $0.9845$ & $0.9745$ & $0.9573$ \\ \hline
	\end{tabular}
	\caption{Capacity of arrays with different aspect ratios.}
	\label{table2}
	\vspace{-1em}
\end{table}

We further investigate how the aspect ratio affects the averaged capacity by simulating arrays with the same number of cells but different aspect ratios. In Table \ref{table2}, with a total of $16384$ cells, the square array (aspect ratio = 1) has the largest averaged capacity and the $4\times 4096$ array, which has the largest aspect ratio, has the lowest averaged capacity. Intuitively, this can be explained by a larger possible cumulative line resistance $nr_w+mr_b$ in an array with a larger aspect ratio. This observation presents a trade-off between the sometimes desired high aspect ratio and a high averaged capacity. 

\section{Optimal Read Threshold}
In Section III.B, we observe that the channel parameters $p_3^{(ij)}$ are $p_4^{(ij)}$ are dependent on the read current threshold $I_{th}$ --- a user defined parameter that can be optimized for lower raw bit-error rate. In this section, we study the optimal threshold for each cell and for an entire array with the goal of reducing the  raw bit-error rate. We choose to optimize the read resistance threshold $R_{th}=\frac{V_r}{I_{th}}$ as it is equivalent to optimizing $I_{th}$ with fixed read voltage $V_r$. We define $R_{th0}$ be the optimal threshold when no line resistance is considered, i.e.,
\begin{equation}
\label{equ:Rth0_opt}
R_{th0}=\underset{R_{th}}{\text{argmin}}\, qQ\left(\frac{\mu_H-\ln(R_{th})}{\sigma_H}\right)+(1-q)Q\left(\frac{\ln(R_{th})-\mu_L}{\sigma_L}\right).
\end{equation}
$R_{th0}$ can be calculated using standard result from estimation theory. $R_{th0}$ is clearly suboptimal when the line resistance is non-negligible. 

\subsection{Optimal Threshold for Each Cell}
One simple read scheme is to use the optimal resistance thresholds for cells at different locations. We call this a different threshold for each cell (DTEC) scheme. Define the optimal threshold for the cell $(i,j)$ to be $R^{ij}_{th}$. With the objective of minimizing $P(Z_{ij}~=Y_{ij})$, and using equation (\ref{equ:p3}) and equation (\ref{equ:p4}), we have:
\begin{equation}
\label{(equ:Rthij_opt)}
R^{ij}_{th}=\underset{R_{th}}{\text{argmin}}\, qQ\left(\frac{\mu_H-\ln(R_{th}-ir_b-jr_w)}{\sigma_H}\right)+(1-q)Q\left(\frac{\ln(R_{th}-ir_b-jr_w)-\mu_L}{\sigma_L}\right).
\end{equation}
Comparing equations (\ref{equ:Rth0_opt}) and (\ref{(equ:Rthij_opt)}), we get
\begin{equation}
\label{equ:opt_thr_cell}
R^{ij}_{th} = R_{th0}+ir_b+jr_w.
\end{equation}
This result is intuitive as we need to shift the threshold to the right in order to compensate for the adverse effect of the cumulative line resistance. 
\subsection{Optimal Threshold for An Array}
Requiring different thresholds for cells in a $m\times n$ array may not be desirable for circuit designers as doing this may require a lot more comparators. In a typical memory design, cells on the same bitline share the same sensing amplifier, so one threshold for each column is a reasonable choice. To further simplify the memory design, one may even use the same threshold for the entire memory array. Therefore, it is of interest to find the optimal threshold that minimizes the averaged BER for an entire array or a sub-array (such as a column). In this subsection, we deal with the optimal threshold for an array first; this result readily generalizes to any sub-array. We call these schemes the same threshold for many cells (STMC) schemes.

With the objective of minimizing $\frac{1}{mn}\sum_{i,j}P(Z_{ij}~=Y_{ij})$, the optimal threshold for an array is defined as
\begin{equation}
\label{(equ:Rtharray_opt)}
\begin{split}
R_{th\_array} = &\underset{R_{th}}{\text{argmin}}\frac{1}{mn}\sum_{i,j}^{m,n}\Bigg[ qQ\left(\frac{\mu_H-\ln(R_{th}-ir_b-jr_w)}{\sigma_H}\right)\\
&+(1-q)Q\left(\frac{\ln(R_{th}-ir_b-jr_w)-\mu_L}{\sigma_L}\right)\Bigg].
\end{split}
\end{equation}

{While it may be possible to heuristically optimize the single parameter $R_{th}$, in practice we wish to provide an analytical solution based on certain approximations and an efficient iterative search algorithm. Analytically,} the optimization problem in (\ref{(equ:Rtharray_opt)}) is hard to solve for  as it involves a summation of Q-functions. We instead replace this objective function with its upper bound and try to minimize this upper bound. Using Jensen's inequality and the fact that the Q-function is concave for a positive argument, we have the following bound:
\begin{equation}
\begin{split}
&\frac{1}{mn}\sum_{i=1}^{m}\sum_{j=1}^{n}\Bigg[ qQ\left(\frac{\mu_H-\ln(R_{th}-ir_b-jr_w)}{\sigma_H}\right)+(1-q)Q\left(\frac{\ln(R_{th}-ir_b-jr_w)-\mu_L}{\sigma_L}\right)\Bigg]\\
\leq&qQ\left(\frac{\mu_H-A}{\sigma_H}\right)+(1-q)Q\left(\frac{A-\mu_L}{\sigma_L}\right),
\end{split}
\end{equation}
where
\[A = \frac{1}{mn}\sum_{i=1}^{m}\sum_{j=1}^{n}\Big[\ln(R_{th}-ir_b-jr_w)\Big].\]
The gap between the two sides of this inequality is small when the Q-functions are close to being linear, which is indeed the case for Q-functions with large arguments. 

We reformulate the problem using the above inequality:
\begin{equation}
\label{(equ:Rtharray_opt2)}
R_{th\_array}=\underset{R_{th}}{\text{argmin}}\,qQ\left(\frac{\mu_H-A}{\sigma_H}\right)+(1-q)Q\left(\frac{A-\mu_L}{\sigma_L}\right).
\end{equation}
Comparing equations (\ref{equ:Rth0_opt}) and (\ref{(equ:Rtharray_opt2)}), the problem becomes of finding $R_{th\_array}$ such that 
\begin{equation}
\label{equ:log_equ}
\ln(R_{th0})=\frac{1}{mn}\sum_{i=1}^{m}\sum_{j=1}^{n}\Big[\ln(R_{th\_array}-ir_b-jr_w)\Big].
\end{equation}
Equation (\ref{equ:log_equ}) is hard to solve since it contains a summation of logarithm functions. We provide both an approximation to this equation that is easier to solve, as well as an iterative algorithm that produces a solution to the original equation.

Approximating each term in the right hand side summation by $\ln(R_{th\_array}-\frac{m+1}{2}r_b-\frac{n+1}{2}r_w)$, the approximate solution of (\ref{equ:log_equ}) can be found:
\begin{equation}
\label{equ:log_equ_sol_1}
R_{th\_array}\approx R_{th\_array\_appx}=R_{th0}+\frac{m+1}{2}r_b+\frac{n+1}{2}r_w.
\end{equation}
This approximation can be also interpreted as averaging of the optimal thresholds, given by equation (\ref{equ:opt_thr_cell}), of all cells.

We propose Algorithm 1 to compute the exact solution of equation (\ref{equ:log_equ}).

\begin{algorithm}[h]
	\caption{STMC threshold solver algorithm}
	\SetAlgoLined
	1. Initialize $R^{(0)}_{th}=R_{th0},{l}=0$.\\
	2. ${l}={l}+1$.\\
	\quad Calculate $R^{({l})}_{th}$ such that $\ln(R^{({l})}_{th}) = \ln(R_{th0})-\frac{1}{mn}\sum_{i=1}^{m}\sum_{j=1}^{n}\ln\left(1-\frac{ir_b+jr_w}{R^{({l}-1)}_{th}}\right)$.\\
	3. Repeat step 2 until a certain iteration is reached or $R^{({l})}_{th}-R^{({l}-1)}_{th}\leq\epsilon$. Let $R_{th\_array} = R^{({l})}_{th}$.
\end{algorithm}

The STMC threshold solver algorithm is inspired by how the summation of logarithm functions is handled in the Expectation Maximization algorithm. The convergence of the STMC threshold solver algorithm is demonstrated in the following lemma.
\begin{lemma}
	The STMC threshold solver algorithm converges to the solution of equation (\ref{equ:log_equ}). 
\end{lemma}
\begin{proof}
	Equation (\ref{equ:log_equ}) can be rewritten as:
	\[\ln(R_{th\_array}) = \ln(R_{th0})-\frac{1}{mn}\sum_{i=1}^{m}\sum_{j=1}^{n}\ln\left(1-\frac{ir_b+jr_w}{R_{th\_array}}\right).\]
	With $R^{(0)}_{th}\geq  R_{th\_array}$, the following inequality holds:
	\[
	\begin{split}
	\ln(R_{th\_array})&\geq \ln(R^{(1)}_{th})= \ln(R_{th0})-\frac{1}{mn}\sum_{i=1}^{m}\sum_{j=1}^{n}\ln\left(1-\frac{ir_b+jr_w}{R^{(0)}_{th}}\right).
	\end{split}\]
	By induction, $\ln(R_{th\_array})$ can be bounded:
	$\ln(R^{(2j+1)}_{th})\leq \ln(R_{th\_array})\leq \ln(R^{(2j)}_{th}), j\in\mathbb{Z}.$
	With {$R^{(2)}_{th}> R^{(0)}_{th}=R_{th0}$}, since $\ln(R^{({l})}_{th}) = \ln(R_{th0})-\frac{1}{mn}\sum_{i=1}^{m}\sum_{j=1}^{n}\ln\left(1-\frac{ir_b+jr_w}{R^{({l}-1)}_{th}}\right)$, the subsequent upper bounds and lower bounds are shrinking toward $\ln(R_{th\_array})$, i.e., $\ln(R^{(2j+2)}_{th})< \ln(R^{(2j)}_{th})$ and $\ln(R^{(2j+3)}_{th})>\ln(R^{(2j+1)}_{th})$. Therefore, the STMC threshold solver algorithm converges.
\end{proof}
{For the simulations we report on in the next section,  the STMC threshold solver algorithm converges in $5$ iterations which is more effective than performing a line search to solve equation (\ref{(equ:Rtharray_opt)}) empirically.}

\subsection{Simulation Results}
We simulate arrays to examine how the averaged read BERs are affected by different read thresholding schemes. Four different thresholding schemes are compared: the naive scheme where $R_{th0}$ is used as the threshold; the DTEC scheme; the STMC scheme with the approximated solution in (\ref{equ:log_equ_sol_1}); and the STMC scheme with the exact solution solved by Algorithm 1. Unless otherwise mentioned, we use parameters from Table \ref{table1}. We first vary the array size and fix $r_b=r_w=30\Omega$ and report the result in Fig. \ref{fig:thr_BER_vs_n}. We also vary the wire resistance and fix $n=m=1024$ and report the result in Fig. \ref{fig:thr_BER_vs_r_wire}. 

\begin{figure}[h]
	\centering
	\begin{subfigure}{0.49\textwidth}
	\centering
	\includegraphics[scale=0.35,clip]{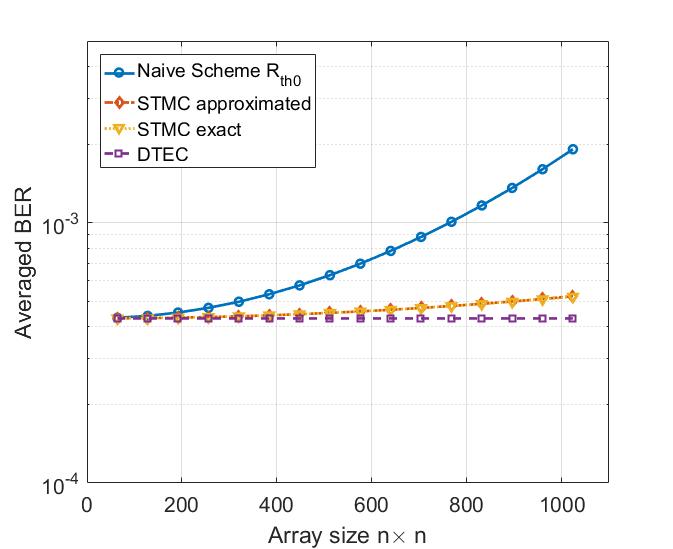}
	\caption{Averaged read BER v.s. array size using different thresholding scheme.}
	\label{fig:thr_BER_vs_n}
	\end{subfigure}
	\begin{subfigure}{0.49\textwidth}
	\centering
	\includegraphics[scale=0.35,clip]{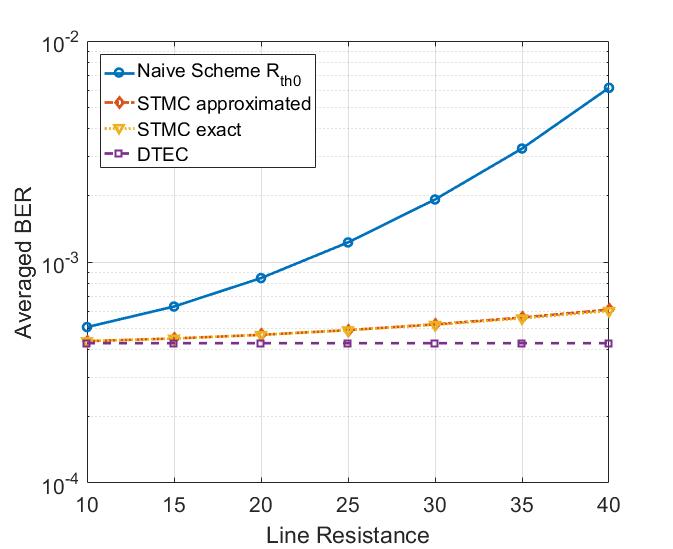}
	\caption{Averaged read BER v.s. wire resistance using different thresholding scheme.}
	\label{fig:thr_BER_vs_r_wire}
	\end{subfigure}
	\caption{Simulation Results Using the Optimal Threshold.}
\end{figure}

From both Fig. \ref{fig:thr_BER_vs_n} and Fig. \ref{fig:thr_BER_vs_r_wire}, we observe that both DTEC and STMC schemes reduce the read BER significantly compared with the naive thresholding scheme. The DTEC scheme compensates for each cell based on its location, and, as a result, the averaged BER is independent of the array size and wire resistance. As expected, the averaged BERs using the exact solution of the STMC scheme is smaller than the averaged BERs using the approximated solution. However, the improvement is incremental and can not be identified on the plots. This shows that equation (\ref{equ:log_equ_sol_1}) approximates the solution of equation (\ref{equ:log_equ}) very well. 

To investigate the effect of the optimized read threshold on the non-uniformity of read reliability, we generate heatmaps depicting the read BER with the unoptimized read threshold (Fig. \ref{fig:heatmap_original_read}) and with the optimized read threshold by the STMC scheme (Fig. \ref{fig:heatmap_optimal_read}). 
\begin{figure}[h]
	\centering
	\begin{subfigure}{0.49\textwidth}
		\centering
		\includegraphics[scale=0.27,clip]{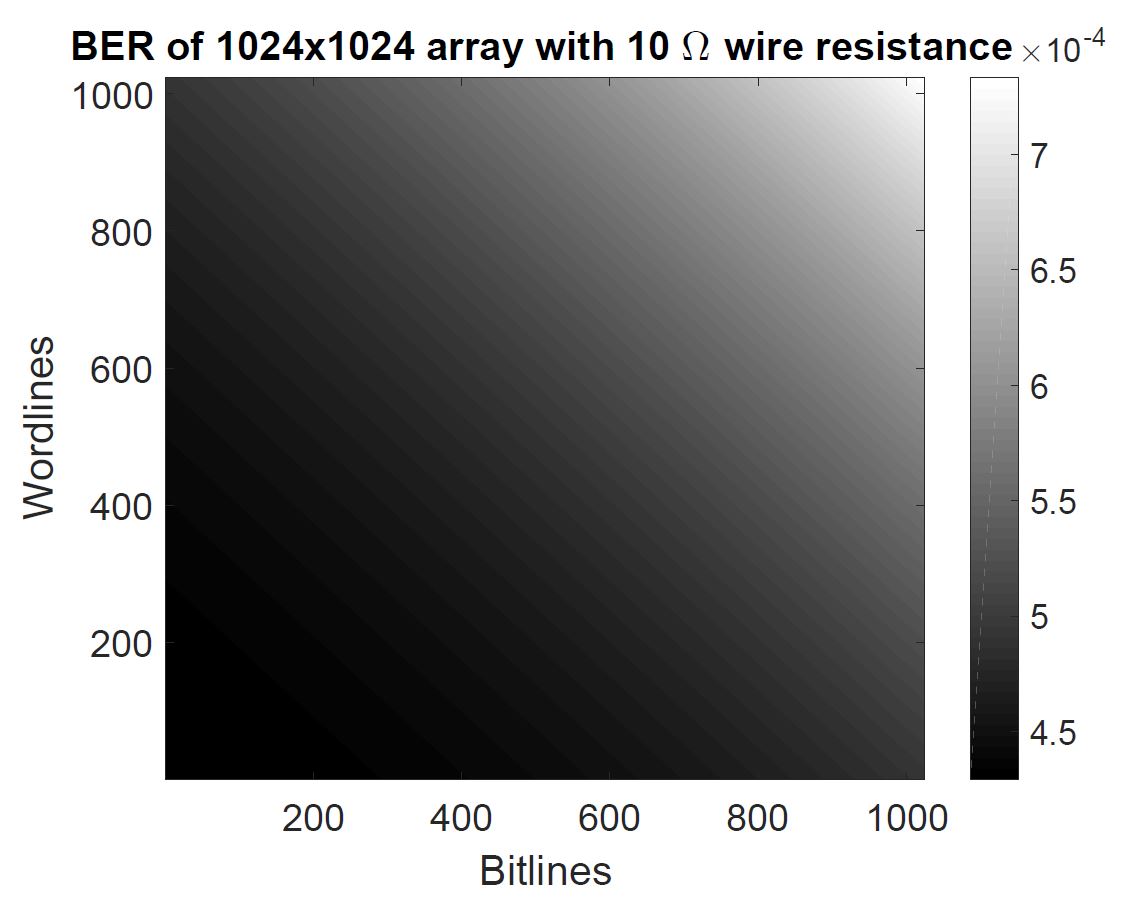}
		\caption{Read BER heatmap with $R_{th0}$.}
		\label{fig:heatmap_original_read}
	\end{subfigure}
	\begin{subfigure}{0.49\textwidth}
		\centering
		\includegraphics[scale=0.27,clip]{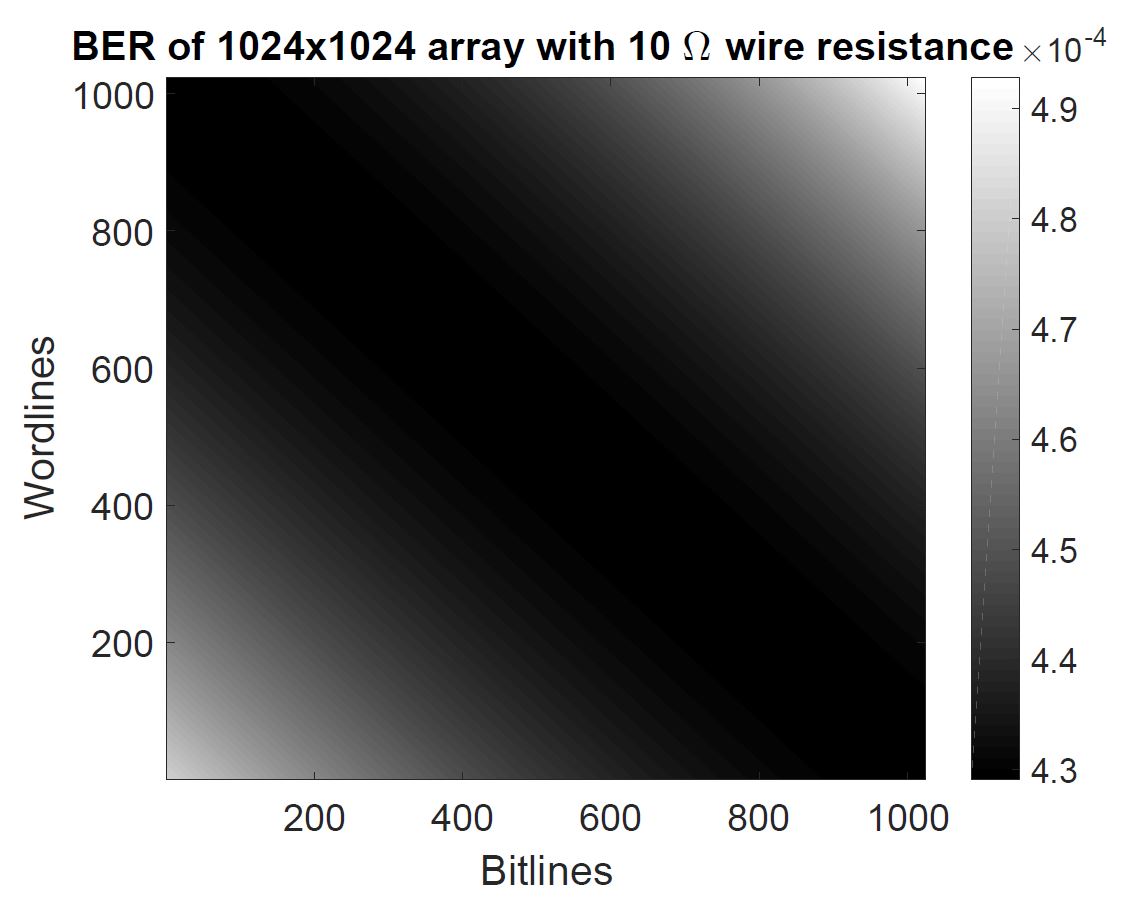}
		\caption{Read BER heatmap with STMC scheme.}
		\label{fig:heatmap_optimal_read}
	\end{subfigure}
%	\begin{subfigure}{0.45\textwidth}
%	\centering
%	\includegraphics[scale=0.45,trim=0.8in 2.8in 1in 2.8in,clip]{new_plots/Wordline_BER_compare}
%	\caption{Averaged BER for each wordline.}
%	\label{fig:read_compare}
%\end{subfigure}
	\caption{Heatmap Comparison Between Unoptimized and Optimized Read Thresholds.}
\end{figure}

Similar to what we observed in Fig. \ref{fig:heat_map} in Section III.D, with the unoptimized threshold, the lower left corner cell is the most reliable one and the upper right corner cell is the least reliable one. With the optimized read threshold, the non-uniformity pattern changes and the reliable cells lie on the diagonal of the array. This result is expected because from equation (\ref{equ:log_equ_sol_1}), we observe that the approximated solution for the STMC scheme is the optimal read threshold for the centermost cell of an array. Note that in our simulation, the device parameters are selected such that the read channel with unoptimized threshold is comparable to the write channel. Because of this set-up, with the optimized read threshold, the read channel is dominated by the write channel and the non-uniformity pattern of the cascaded channel is unchanged. 

\section{Channel Coding for Storage-Class Memory (SCM) Applications}
Error correction codes (ECCs) have become an essential part in memory systems. In an array with largely non-uniform device reliability, efficient ECC solutions must take this non-uniformity into consideration by either mitigating it or leveraging it. Channels with non-uniformity can be considered as channels with SNR variation \cite{esfahanizadeh2018spatially} or non-stationary channels \cite{zorgui2019polar,mahdavifar2020polar}. These models have been studied in the ECC literature in the context of e.g., LDPC codes and Polar codes for storage applications \cite{zorgui2019polar,esfahanizadeh2018spatially,mahdavifar2020polar}. However, these complex codes with long block length and high decoding latency (on the order $\mu s$ \cite{zhao2013ldpc}) are not compatible with the small decoding granularity and fast decoding requirements of the SCM applications. Therefore, in targeting the SCM applications, we study efficient coding schemes for crossbar resistive memory with high line resistance and base them on simple yet effective BCH codes with short block length. BCH codes, which are characterized by the block length ($n$), number of information bits ($k$), and the error correction capability ($t$), have decoding latency on the order of $ns$ and are widely adopted in crossbar resistive memory \cite{niu2012low,choi2018decoder,mao2017multilayer}. 

In subsection V.A, based on a single BCH code, we propose an interleaved coding scheme to mitigate the non-uniformity by allowing cells that store different codewords to have similar averaged RBER. In subsection V.B, we propose a method to use different BCH code in different wordline to leverage the non-uniformity, and do so without interleaving. A systematic framework to optimize the allocations of codes to wordlines is also proposed in V.B. {The two proposed approaches are suitable for different scenarios. The interleaved coding scheme in subsection V.A is conceptually simpler and requires only one pair of ECC encoder and decoder. Meanwhile, the interleaved coding scheme requires interleaver and de-interleaver, which adds additional hardware and latency overhead. The proposed scheme in subsection V.B does not use an interleaving operation but requires an optimization process in the design stage and also requires additional ECC encoders/decoders due to the use of multiple ECCs. Also, because the optimization steps are not typically performed on-the-fly, performance of the proposed scheme based on an optimized allocation can deteriorate if device parameters vary significantly in time. We present both schemes for completeness and the usefulness of each of them in a specific application depends on their performance and the trade-off between hardware/latency overhead of the interleaver/de-interleaver and that of the additional encoders/decoders. Some performance trade-offs between these two approaches will be discussed at the end of subsection V.B when performance of these two schemes are compared.} Note that in the following two subsections, our simulations consider both the read and write channel and are based on parameters in Table \ref{table1}. The read thresholds are optimized using the STMC scheme in Section IV.
\subsection{Single ECC with interleaving}
Suppose a codeword is stored in a wordline in the crossbar memory. When employing the same ECC for all wordlines, the non-uniformity of raw bit-error rates (RBERs), as depicted in Fig. \ref{fig:heat_map}, transforms into the non-uniformity of undetected bit-error rates (UBERs) and thus makes this coding scheme inefficient. If the ECC is designed for an averaged case, e.g., the channel conditions in the middle wordlines, it can be an overkill for the channel conditions in the lower wordlines while being inadequate for the channel conditions in the upper wordlines. One approach to mitigate non-uniformity of channels is interleaving e.g., \cite{esfahanizadeh2018spatially}. Based on the dependency of channel parameters on the row and column indexes of the cell location, we propose a sub-diagonal interleaved coding scheme which is illustrated in Fig. \ref{fig:interleaving}.
\begin{figure}[h]
	\centering
	\includegraphics[scale=0.45,clip]{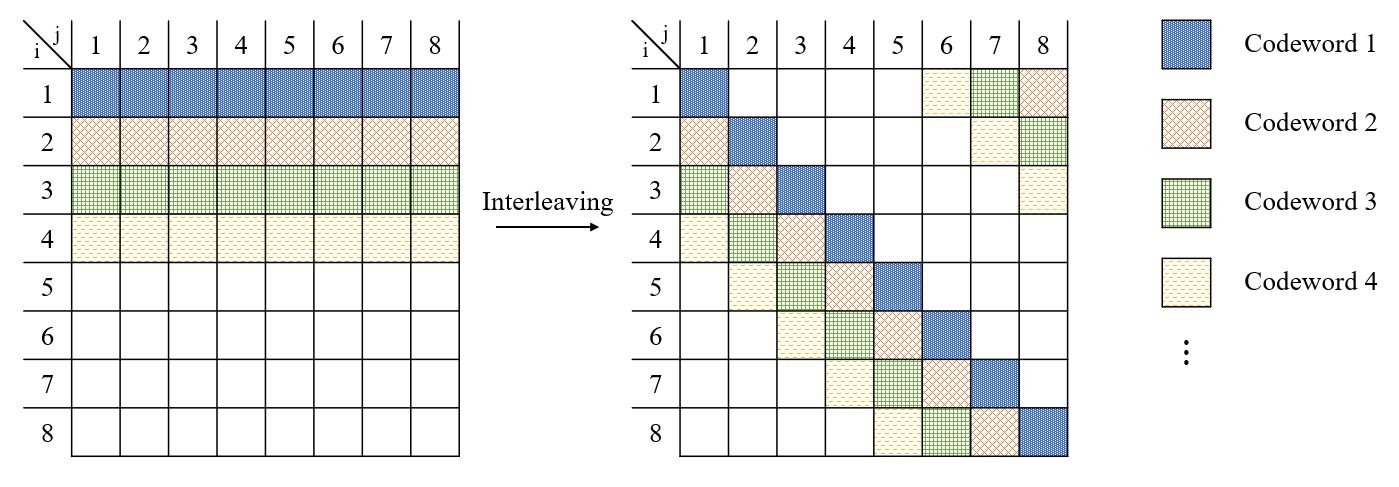}
	\caption{Illustration of the Sub-Diagonal Interleaved Coding Scheme on an $8\times8$ Array.}
	\label{fig:interleaving}
\end{figure}

Based on this sub-diagonal interleaved coding scheme, we store each codeword in cells that are located on the same diagonal or a sub-diagonal of the memory array, instead of cells that are located on the same wordline. Note that this interleaved coding scheme can be readily generalized to a non-square array by storing a codeword in cells located on the main diagonal or a circularly shifted main diagonal of the rectangular array. 
\begin{figure}[h]
	\centering
	\includegraphics[scale=0.3,clip]{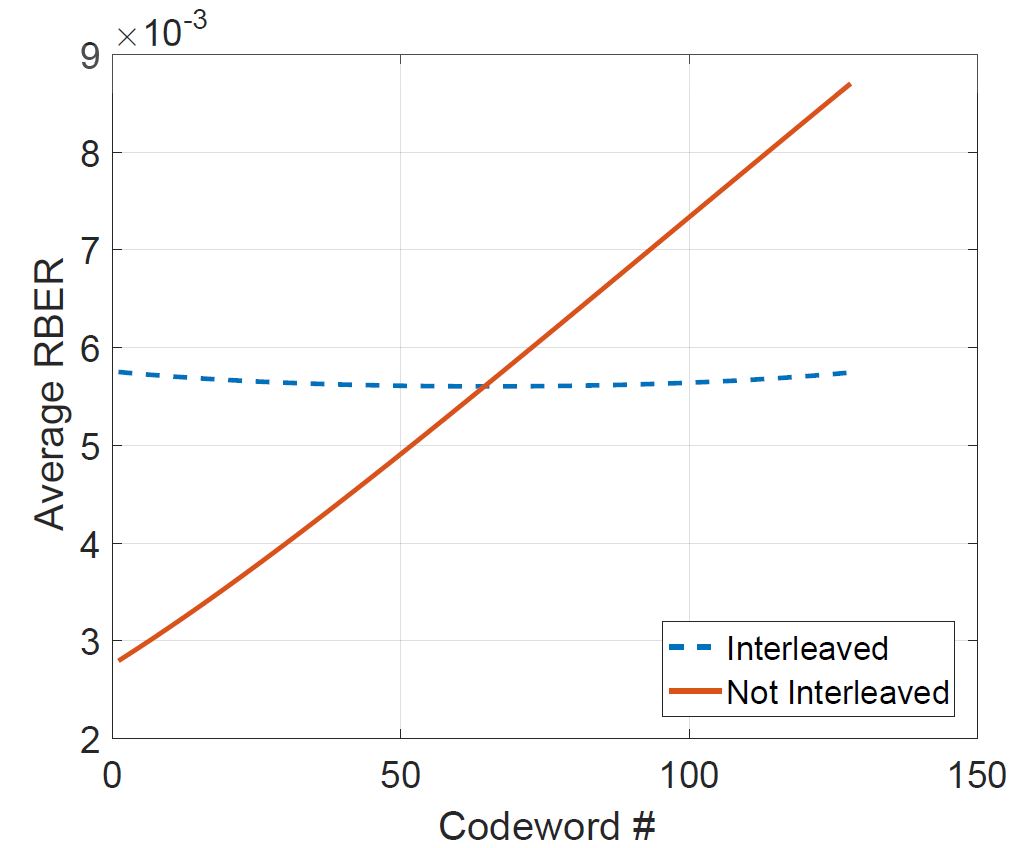}
	\caption{The averaged RBER for Cells  Storing codewords in an $128\times128$ Array with $r_w=r_b=50\Omega$.}
	\label{fig:avgRBER_interleaving}
\end{figure}

To study the effectiveness of this interleaved coding scheme, we first compare the averaged RBER in cells that store each codeword when using the interleaved coding scheme with the averaged RBER in cells that store each codeword where bits of the codeword are stored in a wordline. The results for an $128\times 128$ array are shown in Fig. \ref{fig:avgRBER_interleaving} and we observe that the non-uniformity of averaged RBER among codewords is largely mitigated.

\begin{figure}[h]
	\centering
	\begin{subfigure}{0.49\textwidth}
		\centering
		\includegraphics[scale=0.3,clip]{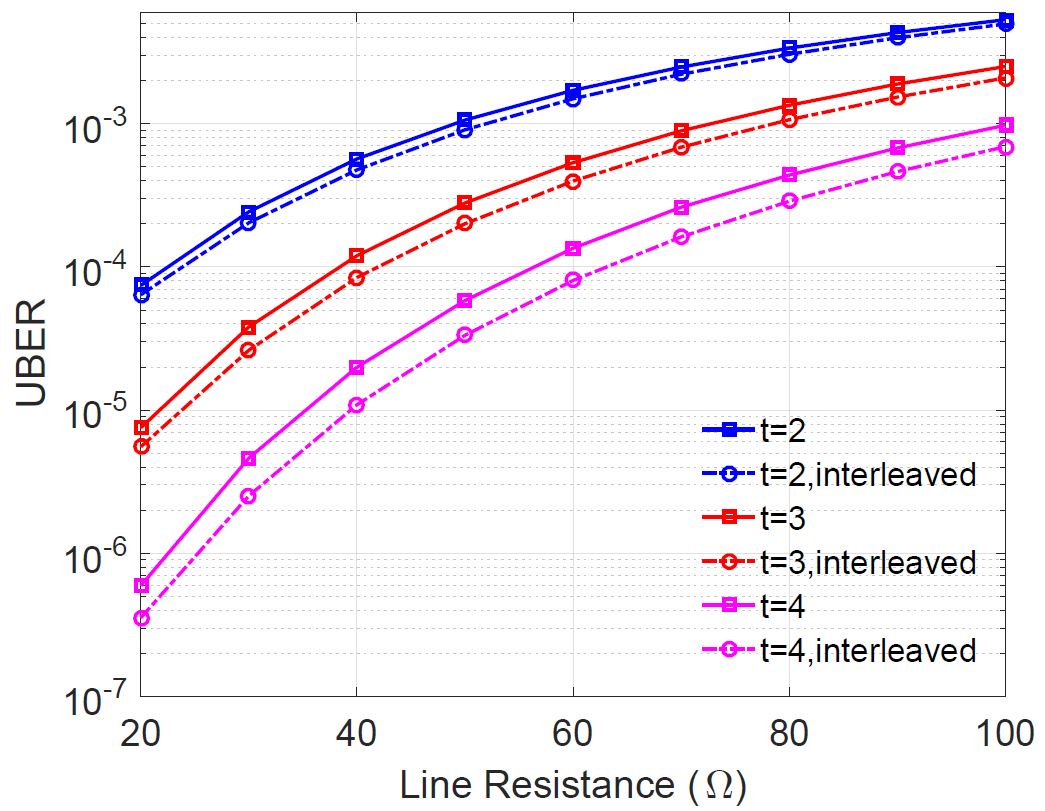}
		\caption{n=128}
		\label{fig:UBER_interleaving_n128}
	\end{subfigure}
	\begin{subfigure}{0.49\textwidth}
		\centering
		\includegraphics[scale=0.3,clip]{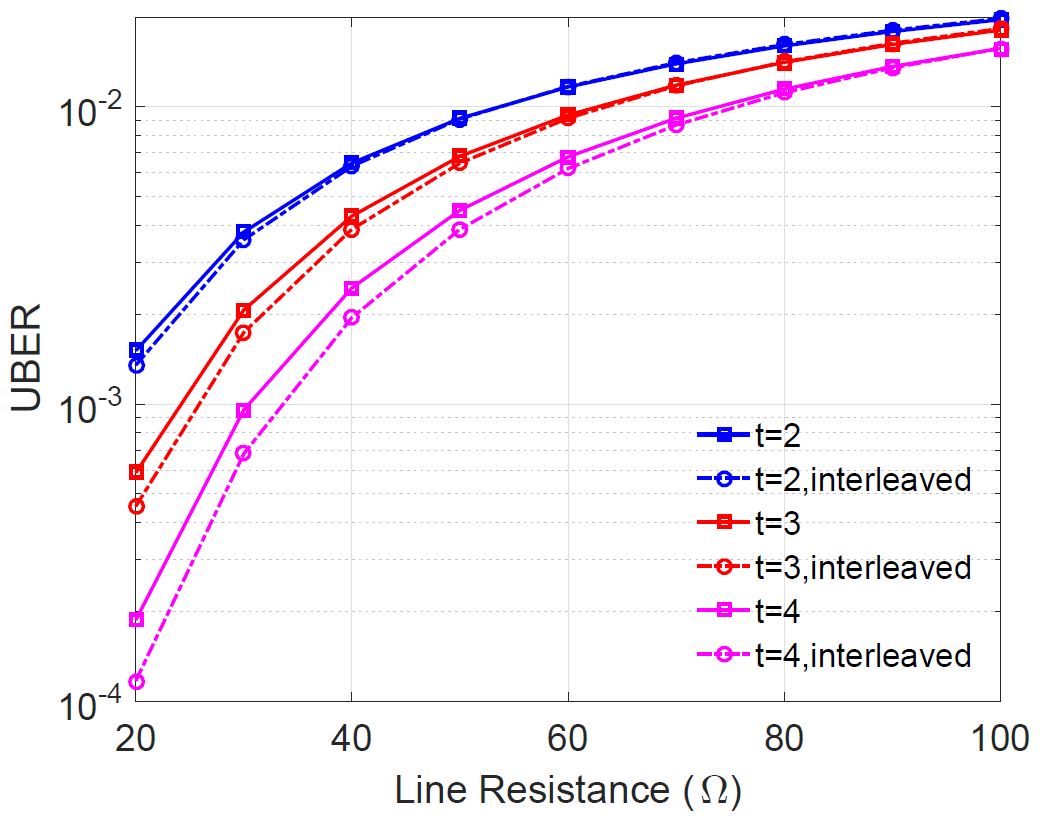}
		\caption{n=256}
		\label{fig:UBER_interleaving_n256}
	\end{subfigure}
	
	\caption{Decoding Performance of the Interleaved Coding Scheme.}
	\label{fig:interleaving_performance}
\end{figure}
We next present the decoding performance of this interleaved coding scheme using BCH codes with block length $n=128$ and $256$, and error correction capability $t=2,3,$ and $4$ in Fig. \ref{fig:interleaving_performance}. In Fig. \ref{fig:interleaving_performance}, the interleaved coding scheme effectively reduces the UBER with a maximum reduction of 45\%. We observe that the interleaved coding scheme is more effective in the regime of short block length and moderate line resistance. We conjecture that even though the interleaved coding scheme is able to mitigate the non-uniformity of the averaged RBER seen by each codeword, it exaggerates the non-uniformity of RBER within a codeword. With a larger block length and line resistance, this trade-off is less favored and the interleaved coding scheme is less effective. We also observe that the UBER improvement is larger with a larger error correction capability. This result is due to the observation that a larger error correction capability gives the interleaved coding scheme more room to balance channels where the code is an overkill  and channels where the code is insufficient, under the non-interleaved case.

\subsection{Multiple ECCs with optimized code allocation}
Although the interleaved coding scheme in the previous subsection is conceptually simple and requires only one ECC, interleaving and de-interleaving are costly operations. Moreover, these operations create additional difficulties if one seeks to utilize the parallel read capability of a memory with the crossbar structure \cite{liang2013effect}. Without interleaving, one intuitive approach to leverage the non-uniform RBERs is to use a strong ECC for relatively unreliable wordlines and use a weak ECC for relatively reliable wordlines. This approach has a reasonable complexity overhead compared to using a single ECC because encoders and decoders for structured codes, e.g., BCH codes, with different error correction capabilities can share the same encoding and decoding circuitry.  Given a set of ECCs with varying error correction capability and a set of wordlines with varying channel parameters, we refer to the problem of designating an ECC to an wordline for all wordlines as the code allocation problem. Similar problems, i.e., using different ECCs for different channel conditions, have been studied in communications under the context of adaptive coding \cite{vucetic1991adaptive} where the channel condition is often time dependent (instead of being location dependent as in our set-up). In this subsection, we present a systematic framework, referred to as the location dependent code allocation (LDCA) framework, for solving the optimal code allocation problem.  We first present the formulation of the optimal location dependent code allocation (LDCA) problem as a standard optimization problem and then propose an effective algorithm to solve for a sub-optimal solution, which is shown to reduce the UBER effectively. 

Suppose we have a set $\mathbb{C}$ of $L$ error correction codes $C_1,C_2,\cdots,C_L$ with the same block length $n$ and different error correction capabilities. In our specific case, we consider BCH codes with a corresponding set of error correction capabilities $\mathbb{T}=\{t_1,t_2,\cdots,t_L\}$.  The memory array is of size $m\times n$ and the crossover probabilities of the cascaded BAC in Fig. \ref{fig:cascaded_channel}, $p_5^{(ij)}$ and $p_6^{(ij)}$, are known based on the channel model. For each wordline $i$, we choose a code from $\mathbb{C}$ and store the encoded bits in it. Let $\bm{C} \in \mathbb{R}^{m\times L}$ be an array whose element $c_{il}$ is the UBER (cost) when code $C_l$ is applied to wordline $i$ and $\bm{c}^T_i$ be a row of $\bm{C}$. Let $\bm{A}\in \mathbb{R}^{L\times m}$ be an array whose column $\bm{a}_i \in \mathbb{R}^L$ is an one-hot vector denoting the code selection for wordline $i$, i.e., for a given $i$, $a_{li}=1$ if $C_l$ is used for wordline $i$ and $a_{li}=0$ otherwise. Let $\bm{r}\in\mathbb{R}^{L}$ represent the rate of the codes in $\mathbb{C}$ and let $R_{goal}$ be the desired storage rate for the memory array. Under the constraint of rate $R_{goal}$, we seek to find a code allocation matrix $\bm{A}$ such that the overall UBER is minimized. The optimal LDCA problem can therefor be formulated as the following integer programming problem:
\begin{equation}
\label{op1}
\begin{split}
\min_{\bm{A}}\quad&\sum_{i=1}^{m}\bm{c}_i^T\bm{a}_i=Tr(\bm{AC})\\
\text{s.t.}\quad&\bm{a}_i^T\bm{1}=1, \quad i=1,\cdots,m\\
&a_{li}\in\{0,1\}\\
&\frac{1}{m}\sum_{i=1}^{m}\bm{r}^T\bm{a}_i\leq R_{goal}
\end{split}
\end{equation}

In the above optimization problem, $\bm{1}$ is the all-ones vector. To solve the above optimization problem, we first need to estimate the cost matrix $\bm{C}$ from the channel parameters. We calculate $\bm{C}$ based on the following approximation:
\begin{equation}
\label{equ:approximation_for_c}
c_{il}\approx1-\sum_{e=0}^{t_l}\binom{n}{e}(\bar{p}_i)^{e}(1-\bar{p}_i)^{n-e},
\end{equation} 
where
\[
\bar{p}_i=\frac{1}{n}\sum_{j=1}^{n}\left[qp_5^{ij}+(1-q)p_6^{ij}\right].
\]
In (\ref{equ:approximation_for_c}), we approximate the $n$ BAC channels of cells on a wordline by an averaged binary symmetric channel (BSC) with parameter $\bar{p}_i$ and calculate the analytical UBER for a BCH code based on this approximation. 

With the approximated cost matrix $\bm{C}$, we propose Algorithm 2 to solve the integer programming problem (\ref{op1}) by solving its relaxed linear programming problem (\ref{op2}) iteratively:
\begin{equation}
\label{op2}
\begin{split}
\min_{\bm{A}}\quad&\sum_{i=1}^{m}\bm{c}_i^T\bm{a}_i=Tr(\bm{AC})\\
\text{s.t.}\quad&\bm{a}_i^T\bm{1}=1, \quad i=1,\cdots,m\\
&0\leq a_{li}\leq 1\\
&\frac{1}{m}\sum_{i=1}^{m}\bm{r}^T\bm{a}_i\leq R_{LP}
\end{split}
\end{equation}

\begin{algorithm}[h!]
	\caption{LDCA solver algorithm}
	\SetAlgoLined
	1. Solve the linear programming problem (\ref{op2}) with $R_{LP}=R_{goal}$ to get $A^*$.\\
	2. Estimate the code allocation matrix $\hat{A}$ based on the maximum in each column of $A^*$. Namely, in $\hat{\bm{a}}_i$, only the $\text{argmax}(\bm{a}^*_i)$-th element is set to $1$ and all other elements are $0$.\\
	3. Calculate the true rate based on $\hat{A}$, i.e., $\hat{R}=\frac{1}{m}\sum_{i=1}^{m}\bm{r}^T\bm{a}_i$.\\
	4. Terminate the algorithm if $\hat{R}$ is within a certain tolerance of $R_{goal}$. Otherwise, update $R_{LP}$. Namely,\\
	\Indp
	\uIf{ $R_{goal}-R_{tol}\leq \hat{R}\leq R_{goal}+R_{tol}$}{
		Algorithm terminates;
	}
	\uElseIf{$\hat{R}>R_{goal}+R_{tol}$}{
		Update $R_{LP}$ with $R_{LP}-R_{update}$;
	}
	\Else{
		Update $R_{LP}$ with $R_{LP}+R_{update}$;
	}
	\Indm
	5. Repeat step 1-4 until termination or certain number of iterations is reached.
\end{algorithm}

In Algorithm 2, $R_{tol}$ and $R_{update}$ are user defined parameter specifying the tolerance between the true rate $\hat{R}$ and the target rate $R_{goal}$, and the step for updating the rate used in the linear programming problem, respectively. Note that in our experiment with $L=3$, Algorithm 2 terminates in a single iteration.

To better illustrate the effectiveness of Algorithm 2, we present rows of the solved code allocation matrix $A^*$ for an array with $n=m=128$ and varying line resistance in Fig. \ref{fig:Algorithm2}. The channel parameters are computed based on our model in Section III and device parameters in Table \ref{table1}. We include $L=3$ codes in $\mathbb{C}$: BCH(128,100,3), BCH(128,93,4) and BCH(128,86,5) and set $R_{goal}$ to be the rate of BCH(128,93,4). Based on the definition of $\bm{A}$, each row of $A$ can be interpreted as the relative ``weight'' of each code. We observe that with the $10\Omega$ line resistance, the non-uniformity is mild and the solution from Algorithm 2 suggests that using a single code BCH(128,93,4) is sufficient. As the line resistance goes up, the solution from Algorithm 2 suggests the usage of a strong code BCH (128,86,5) for the worse wordlines and a weak code BCH(128,100,3) for the better wordlines. 
\begin{figure}[h]
	\centering
	\includegraphics[scale=0.45,clip]{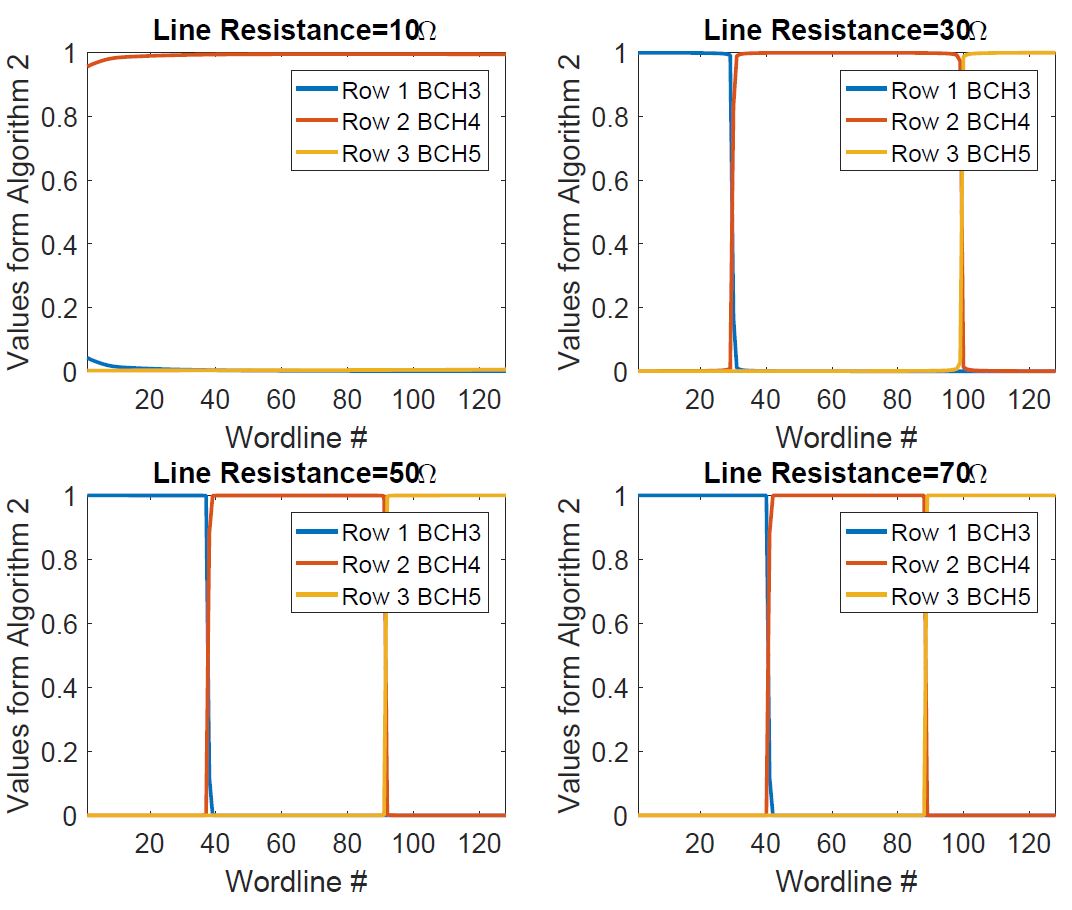}
	\caption{Solutions of Algorithm 2 for Arrays with Various Line Resistance Values.}
	\label{fig:Algorithm2}
\end{figure}

We perform simulations to study the effectiveness of our proposed LDCA framework. In our simulations, two block lengths $n=128$ and $n=256$ are studied, and the corresponding memory arrays are of sizes $128\times128$ and $256\times256$, respectively. Line resistances are varied while the remaining device parameters are from Table \ref{table1}. For each set of the array parameters, four different sets of codes are tested. These sets of codes are the standard BCH codes with $\mathbb{T}=\{1,2,3\}$, $\mathbb{T}=\{2,3,4\}$, $\mathbb{T}=\{3,4,5\},$ and $\mathbb{T}=\{1,2,3,4,5\}$. Their target rates are set to equal the rate of the BCH codes with $t=2,t=3,t=4,$ and $t=3$, respectively. For each group of codes, the optimized code allocation is first solved for by the LDCA solver algorithm, and then the UBER based on the optimized code allocation is simulated. We compare the performance of each group of codes with the performance of the BCH code of the same rate with and without the interleaved coding scheme. The results are shown in Fig. \ref{fig:LDCA_performance}.
\begin{figure}[h]
	\centering
	\begin{subfigure}{0.49\textwidth}
		\centering
		\includegraphics[scale=0.3,clip]{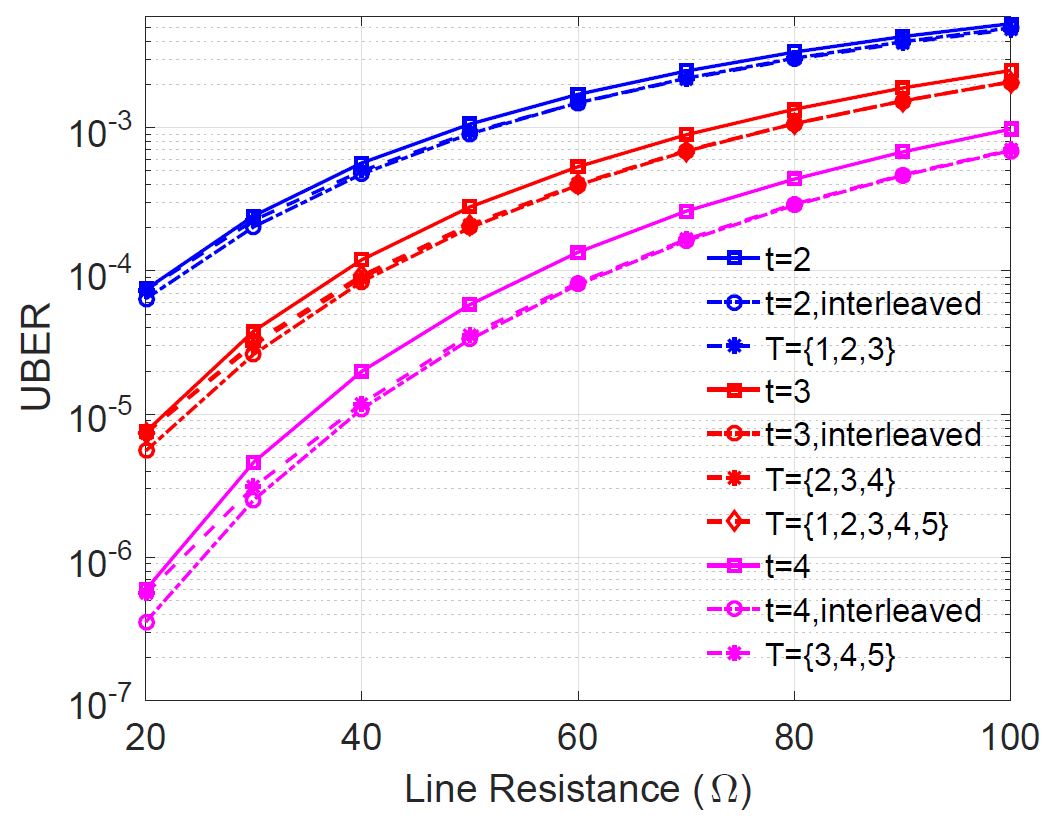}
		\caption{n=128}
		\label{fig:UBER_LDCA_n128}
	\end{subfigure}
	\begin{subfigure}{0.49\textwidth}
		\centering
		\includegraphics[scale=0.3,clip]{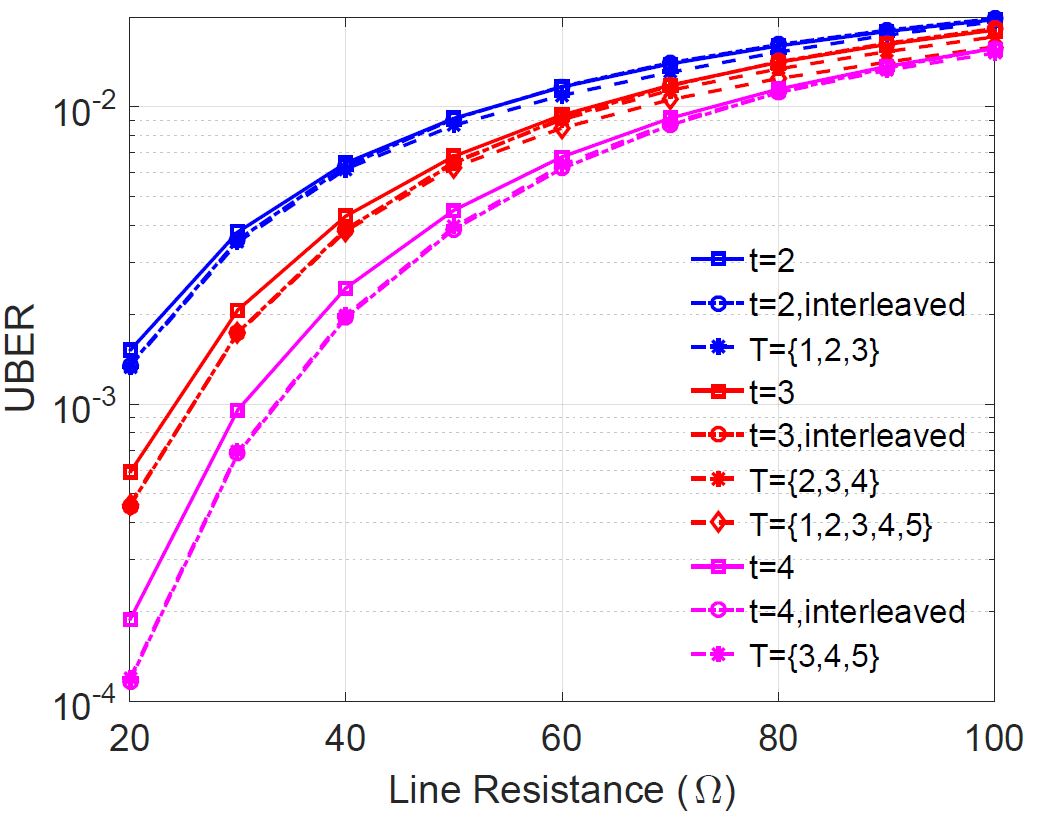}
		\caption{n=256}
		\label{fig:UBER_LDCA_n256}
	\end{subfigure}
	
	\caption{Decoding performance of sets of codes based on LDCA framework}
	\label{fig:LDCA_performance}
\end{figure}

From Fig. \ref{fig:LDCA_performance}, we observe that the scheme utilizing a set of codes based on our LDCA framework shows consistent improvement compared to the scheme utilizing a single code. The UBER of a set of $L=3$ codes is comparable to the UBER of a single code of the same rate with interleaving. In the $128\times128$ array, the set of codes with $\mathbb{T}=\{1,2,3,4,5\}$ does not outperform the set of codes with $\mathbb{T}=\{2,3,4\}$. Meanwhile, in the $256\times256$ array, the set of codes with $\mathbb{T}=\{1,2,3,4,5\}$ outperforms the set of codes with $\mathbb{T}=\{2,3,4\}$ by a consistent $7\%$ in the high line resistance regime. This observation suggests that the utility of allowing for more ECCs to be used in the system depends on the severity of non-uniformity in the memory array. Comparing the proposed schemes in subsection V.A and that in subsection V.B, we observe that the decoding performance of the interleaved coding scheme and the decoding performance of the optimized sets of $L=3$ codes are very similar in most of the tested operating regime. In the regime of short block lengths, i.e., $n=128$, and low line resistance, the interleaved coding scheme outperforms the scheme based on multiple ECCs. We conjecture that this is because interleaving is able to mitigate the non-uniformity despite its severity, whereas the scheme based on multiple ECCs requires notable RBER differences among the wordlines to perform well as the error correction capabilities of ECCs are only expressed in discrete integer values. For larger block lengths, i.e., $n=256$, and high line resistance, the scheme based on the set of $L=5$ codes shows notable improvement compared to the interleaved coding scheme. This shows that when allowing for more ECCs to be used in the LDCA framework, a scheme based on the optimized code allocation of a larger collection of codes is more promising than the interleaved coding scheme when the non-uniformity across the memory array is severe.

Note that as the effectiveness of the scheme utilizing the proposed LDCA framework depends on the code length and severity of non-uniformity, other device parameters from Table \ref{table1} also impact the performance of the proposed scheme. Because the main focus of this paper is on the effect of line resistance, we only present results with varying line resistance and the study on the effect of other device parameters are left for future work. Also note that although our simulations is based on BCH codes and our proposed channel models, the LDCA framework is a general one and it is applicable to other systems with location dependent non-uniformity even when bearing a different channel model and/or utilizing different ECCs. The LDCA framework can be readily generalized to other channels models and ECCs by changing how the cost matrix $\bm{C}$ is calculated in Equation (\ref{equ:approximation_for_c}). 

If one seeks to take into consideration the cost of decoders with different error correction capability or if one seeks to find solutions that use fewer codes without manually restricting $L$, a simple variant of the LDCA framework can be used by adding a regularization term $\bm{c}_{dec}^T\bm{A}\bm{1}$ in the objective functions of the optimization problems where $\bm{c}_{dec}\in\mathbb{R}^{L}$ is a vector that contains appropriate weights for the decoders. To demonstrate the effectiveness of this regularized variant, we present an example by generating plots similar to Fig. \ref{fig:Algorithm2} with and without this regularization term in Fig. \ref{fig:LDCA_reg_results}. For an $256\times 256$ array with $r_w=r_b=30\Omega$, the solution from the LDCA solver algorithm (Fig. \ref{fig:LDCA_unreg}) suggests us to use the BCH(256,248,1) code for the first few wordlines and the BCH(256,216,5) for the last few wordlines. In this moderate line resistance regime, one may wish to avoid using BCH(256,248,1) and BCH(256,216,5) as observation in Fig. \ref{fig:LDCA_performance} suggests that improvement of using more codes is incremental when the non-uniformity is moderate. Instead of manual disallowing those two codes in $\mathbb{C}$, an alternative is to use the regularization term  $\bm{c}_{dec}^T\bm{A}\bm{1}$. With this regularization term, we get a solution that is shown in Fig.\ref{fig:LDCA_reg} and this solution does not use BCH(256,248,1) and BCH(256,216,5). This regularized variant is therefore capable of automatically penalizing the use of more ECCs when the non-uniformity is moderate while still allowing more ECCs to be used when the non-uniformity is severe. 
\begin{figure}[h]
	\centering
	\begin{subfigure}{0.49\textwidth}
		\centering
		\includegraphics[scale=0.3,clip]{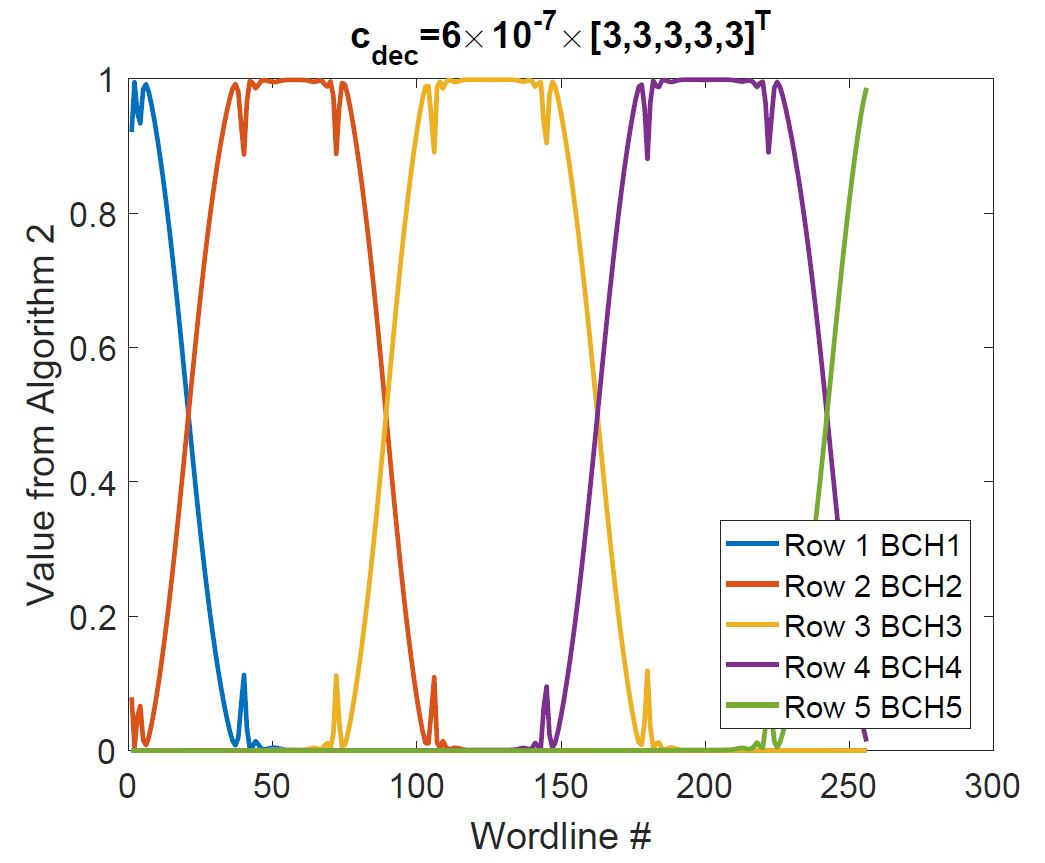}
		\caption{un-regularized solution}
		\label{fig:LDCA_unreg}
	\end{subfigure}
	\begin{subfigure}{0.49\textwidth}
		\centering
		\includegraphics[scale=0.3,clip]{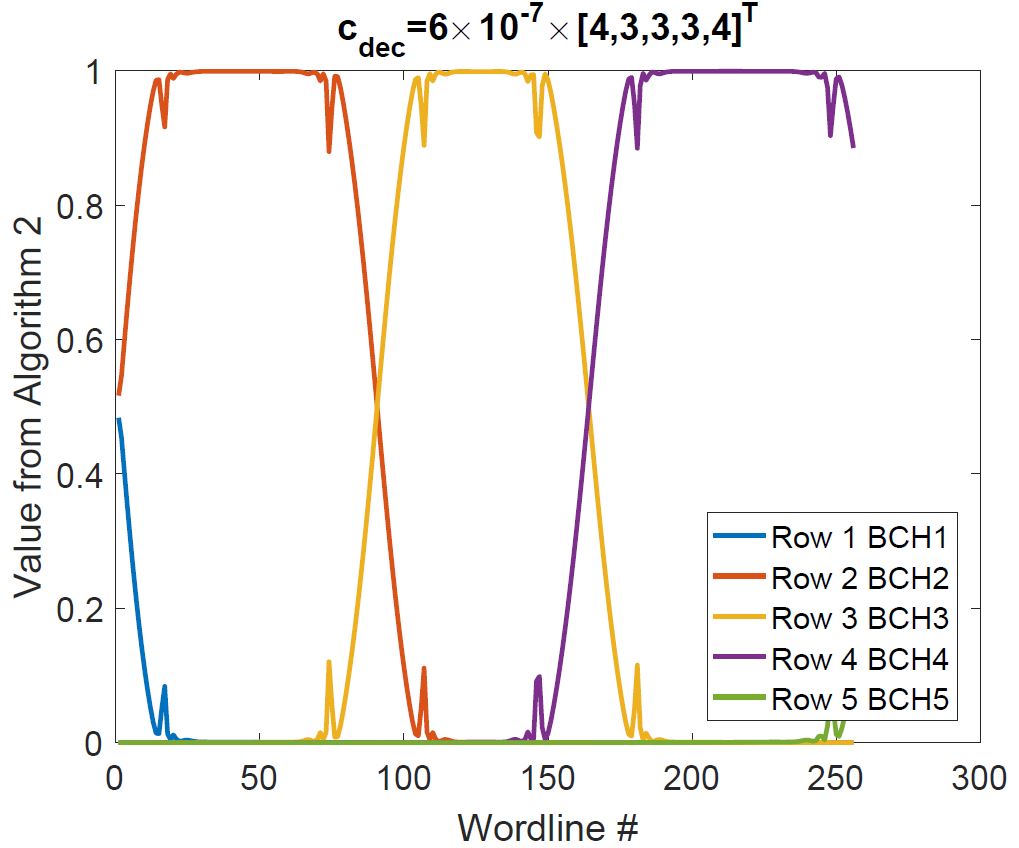}
		\caption{regularized solution}
		\label{fig:LDCA_reg}
	\end{subfigure}
	
	\caption{An example of using the regularization term in the LDCA framework.}
	\label{fig:LDCA_reg_results}
\end{figure}
\section{Conclusion and Future Works}
In this paper, we proposed the write and read channel models for the 1S1R crossbar resistive memory while considering the nondeterministic nature of the memory device. We studied the optimal read threshold which reduces the RBER efficiently. Two schemes, one that utilizes interleaving and one that utilizes multiple codes based on a proposed location dependent code allocation (LDCA) framework, are proposed to improve UBER performance. Future research includes extensive evaluation of the proposed coding scheme under different device parameters.


\end{document}